\documentclass[11pt,reqno]{amsart}  
\usepackage{amsaddr}

\usepackage{xcolor}

\RequirePackage{amsthm,amsmath,amssymb}
\RequirePackage{hyperref}
\usepackage{cancel, color}

\usepackage{etoolbox}     
\makeatletter
\patchcmd{\appendix}{\@Alph}{\@Roman}{}{}
\makeatother

\usepackage{parskip}
\usepackage{enumitem}

\hypersetup{
    linktocpage,
    linkcolor={black} 
    }
\RequirePackage{tabularx}

\theoremstyle{plain}
\newtheorem{thm}{Theorem}[section]

\renewcommand{\appendix}{\par   
  \setcounter{section}{0}
  \setcounter{subsection}{0}
  \gdef\thesection{\Alph{section}}
}

\usepackage{empheq}               
\usepackage[strict]{changepage}

\usepackage{bm}
\usepackage{enumitem}

\DeclareMathOperator{\sgn}{sgn}

\usepackage{mathtools}

\def\Z{\mathbb{Z}}                          
\def\R{\mathbb{R}}  
\def\G{\mathbb{G}}
\def\V{\mathcal{V}}
\def\E{\mathcal{E}}

\newtheorem{theorem}{Theorem}[section]

\newtheorem{lemma}[theorem]{Lemma}

\newtheorem{definition}[theorem]{Definition}

\newcommand{\be}[1]{\begin{equation}\label{#1}}
\newcommand{\ee}{\end{equation}}
\numberwithin{equation}{section}

\newcommand{\ba}[1]{\begin{align}\label{#1}}
\newcommand{\ea}{\end{align}}
\numberwithin{equation}{section}

\newcommand{\ben}{\begin{equation*}}
\newcommand{\een}{\end{equation*}}
\numberwithin{equation}{section}

\newcommand{\bbE}{\mathbb{E}}

\newcommand{\bbP}{\mathbb{P}}

\newcommand{\bbR}{\mathbb{R}}

\def\Z{\mathbb{Z}}                          
\def\X{\mathbb X}
\def\be{\begin{equation}}
\def\ee{\end{equation}}

\def\S{{\mathcal S}}

\def\id{\bm{1}}
\def\ind{\bm{1}}
\def\C{\mathbb{C}}
\def\Q{\mathfrak{g}}

\newcommand{\n}{\mathbf{n}}
\newcommand{\m}{\mathbf m}
\newcommand{\om}{\mathbf \omega}

\newcommand{\rk}[1]{\bgroup\color{red}%
  \par\medskip\hrule\smallskip%
  \noindent\textbf{#1}%
  \par\smallskip\hrule\medskip\egroup}

\date{printed \today}

\title{Geometric analysis of Ising models, part III \vspace{-0.3cm}} 

\author{Michael Aizenman$^*$}
\address{Departments of Physics and  Mathematics, Princeton University, USA 
\footnote{$^*$ Weston Visiting Professor at the Weizmann Institute of Science, Israel. \\  
\indent E-mail:  aizenman@princeton.edu
}}

\begin{document}

	\setlength{\parskip}{1ex} 
\setlength{\parindent}{25pt}
		
\begin{abstract}  \mbox{}  \\[2ex]  
The random current representation of the Ising model, along with a related path expansion, has been a source of insight on the  stochastic geometric underpinning of the ferromagnetic model's phase structure and critical behavior in different dimensions.   This representation is 
extended here to systems with a mild amount of frustration, such as generated by disorder operators and external field of mixed signs.  
Further examples of the utility of such stochastic geometric representations are presented in the context of the deconfinement transition of the $\Z_2$ lattice gauge model -- particularly in three dimensions-- and in  streamlined proofs of correlation inequalities with wide-ranging applications.

 \end{abstract}
	
	\maketitle
\tableofcontents

\section{Introduction}{}

Since its introduction some hundred years ago~\cite{Len20,Isi25}, and especially after its validation in  the ground breaking contributions of Peierls~\cite{Pei36}, and Onsager~\cite{Ons44}, the Ising model has been the source of many  insights on symmetry breaking, phase transitions,  and critical phenomena.   The  simplicity of its formulation, the ever increasing range of applicable tools, and the observed principles of universality, have made its study into a worthwhile and rewarding enterprise.  The role it has played for statistical mechanics is comparable to that of Drosophila for genetics~\cite{ShiWei81}.

This  article is not intended to offer a thorough review of it subject.  
What is presented here is a delayed and correspondingly re-edited Part III  of the author's 1982  work on Ising and related models, of which Parts I and II appeared bundled together in~\cite{Aiz82}.   

``Part I''  exposed a stochastic geometric perspective on the structure of spin correlations  in ferromagnetic Ising models.  Its highlight was an intuitive, yet fully rigorous, explanation of some of the dimension dependence on the Ising model's critical behavior.  More specifically, the fact that  its  critical exponents and structure of correlation functions in dimensions $d>4$ are similar to those of the model's mean-field version  while that is not the case  in dimensions $d<4$, for reasons which are particularly transparent in  the planar case of $d=2$.  

``Part II''  was directed at the implications of the above for the constructive field theory project, whose declared goal has been the construction of a non-trivial $\phi^4$ field theory in $d$ dimensions.    For that, the random current  perspective, and the parallel work~\cite{Fro82}, produced contrasting implications concerning the nature of the scaling limits of the order parameter's fluctuation field which was shown to be Gaussian in high dimensions ($d>4$, and possibly also $d=4$ ~\cite{AizGra83, Fro82}) and not so in  dimensions $d<4$, of which  $d=2$ is a particularly transparent case.  Since then, the partly similar stochastic geometric representation  of~\cite{BryFroSpe82}  was successfully applied by Brydges-Fr\"ohlich-Sokal~\cite{BryFroSok83} to the (non-Gaussian) $\phi^4$ field in $d=3$ dimensions.  For the borderline dimension $d=4$ it was eventually 
 proved, in collaboration with H. Duminil-Copin~\cite{AizDum21},  that the natural 
Ising and $\phi^4$ lattice fields have only Gaussian limits  there as well (as many, though not all, expected\footnote{The dust has still not settled over attempts to make sense of theories like $\phi^3_d$  which despite the luck of a probabilistic interpretation seems to have predictive power, e.g. for percolation.  Understanding that may offer another angle on $\phi^4_4$.}).  
  
The present ``Part III'', in line with its originally published outline  focuses on extensions of the tools that were deployed in ~\cite{Aiz82}, and a demonstration of their relevance for a number of other topics. 

Since the publication of   I \&  II, random-current based methods have been successfully extended and applied  to a range of non-perturbative results on the  critical behavior in ferromagnetic Ising spin systems and related models.  Partial lists of those can be found in~\cite{Dum16,Dum17,Aiz22}.  
Some of the extensions which were originally listed for part III did not yet seem to have been sufficiently publicized.  Though  some were  reported in other works, with references to this paper~\cite{Gra84,DunSch25}.  
These include: 
\begin{enumerate}  
\item [(a)]  a probabilistic extension of the random current representation beyond the frustration-free case 
\item [(b)]  the synthesis of RCR with reflection symmetry.
\end{enumerate}
  
For systems with couplings of mixed signs the expansion which yields the random current representation (RCR) produces terms which are no longer positive.   In Section~\ref{sec:frustration} it is shown that also in that case a probabilistic perspective can be recovered.  The key step for that is the identification  through  the switching lemma of exact cancelations which resolve the original RCR's  sign-problem.  
    
The technique is of relevance  the lattice $\Z_2$-gauge model, which is discussed here with particular focus on three dimensions, making use of its duality there with the regular Ising model.
Tools which are enabled through  the positivity mentioned in (a) are used here to  complete the proof of sharpness of the de-confinement transition.   
 More explicitly:   the existing proof  by Lebowitz and Pfister \cite{LebPfi81} of Wilson loop area-law confinement  at  
$\beta < \beta_c \equiv [\beta_c^{(\text{Ising})}]^*$, is supplemented here with a  perimeter-law lower bound on the Wilson loop's expectation  for  $\beta > \beta_c$.

Another application of (a) is a streamlined random-current-based derivation of the recent inequality of Ding-Song-Sun~\cite{DSS23}.   Their new result has facilitated  a number of recent developments, among which is the proof of cutoff in the Ising model's Glauber dynamics at all $T>T_c$~\cite{BauDeg24}.  To this list was just added the  proof that 
the bipartite entanglement in the ground states of the quantum Ising model on $\Z^d$ obeys a variant of an area law throughout its subcritical regime~\cite{AizWar25}. 
(Previous results on quantum ground states' entanglement have addressed the case $d=1$, and states sufficiently below the quantum phase transition threshold \cite{GOS20}).

 In (b) are included examples of conditional switching symmetry in the context of: i) a single configuration folded onto itself along a symmetry hyperplane, and ii) switching between graphs of only partial overlap (though examples of that have already appeared in published works).  
The former yields  an alternative 
proof of the  van Beijeren's theorem~\cite{Bei75}  that at temperatures  below $T_c(\Z^{d-1})$ 
the Dobrushin boundary conditions produce translation symmetry breaking in the nearest neighbor Ising model on $\Z^d$,  
The RCR perspective on the subject has  been cited to this work,  
and discussed more thoroughly 
in Ross Graham's paper~\cite{Gra84}, where the exactness of the link between the detouchment of the RCR-null surface and the  unidirectional translation symmetry breaking is established.   Since the question of existence of a roughening transition in $3D$ (at some $T< T_c(d)$ is still open, the topic remains included here. 

Also in line with the original plan, in the Appendix are presented  RCR-based derivations of few other Ising model  inequalities that were originally derived by different methods.   

An addition which was not anticipated in this article's originally presented TOC is the Fortuin-Kasteleyn  random cluster representation (Section~\ref{sec:FK}).  While the broader class of FK models  was introduced some ten years earlier, at the time its relevance for the study of the critical behavior of the Ising model was still to be realized (\cite{ACCN88}).   As the FK random cluster representation shares some grounds with the RCR random current representation, and on occasions it helps to combine the two (cf. \cite{DGR20}), the relation of the two is commented upon here in Section~\ref{sec:FK}, and further below.

In publishing a delayed work a natural question is ``why now''?  In this case this is answered by the recent surge of interest in the subject,  the relevance of the methods presented here for some recent  works,  
and potentially also for some of the still open questions.  

It should however be stressed that this is not a report on all that has happened since Parts I \& II in the field of rigorous studies of the statistical mechanics of Ising models.  In particular not addressed here are  disorder effects, such as found in spin glass and random field Ising models which are surely mentioned elsewhere in this volume.

\section{The Ising model}\label{sec:Ising} 

By the term Ising model we refer here to any system of $\pm1$ valued``spin'' variables $\{\sigma_x\}_{x\in \mathcal V}$ associated with the sites of a graph $\G = (\mathcal V,\mathcal E)$ (described this way in terms of vertex and edge sets), with an interaction Hamiltonian of the form 
\be \label{H_w_h}
 {H}(\sigma) := - \sum_{\{x,y\} \in \E} J_{x,y} \sigma_x \sigma_y 
    -  \sum_{x \in \V} h_x  \sigma_x 
\,, \ee
Here $ J_{x,y}$ and $h_x$ are the model's control parameters: interactions across the graph's edges, and the external ``magnetic field'' applied at the graph's sites.

The Gibbs equilibrium state of such a system in a finite subset $\Lambda \subset  \mathcal V$,     
at inverse temperature $\beta$ is the probability distribution under which the expectation values of functions $F: \{-1,+1\}^{\Lambda}$ are given by
\be \label{Gibbs}
\langle F \rangle_{\beta,\Lambda}   =  \frac{1}{2^{|\Lambda|} }  \sum_{\sigma\in \{-1,+1\}^{\Lambda}} F(\sigma) e^{-\beta H(\sigma)}/Z_\Lambda(\beta)  \, 
\ee 
which is normalized by  the partition function:  
\be \label{Z} 
Z_{\Lambda,\beta}(J,h) :=   \frac{1}{2^{|\Lambda|} } \sum_{\sigma\in \{-1,+1\}^{\Lambda}}   e^{-\beta H(\sigma)} 
\, . 
\ee

The infinite volume Gibbs  states are defined through limits of the Gibbs states in finite domains, possibly under some specified boundary conditions.   They can also be characterized as solutions of the corresponding local DLR equations. 
Among the first questions of interest are: 
\begin{enumerate}
\item[i)] for what range of $\beta$ (at given $(\mathcal J, h)$) does the limiting state depend on the boundary conditions,  
\item[ii)] what is the decay rate of spin-spin correlations,   
\item[iii)] what is the structure  of the multi-spin correlations, especially at the model's  critical  points.  
\end{enumerate}

 \section{Selected provable results}  

 In the absence of exact solutions, which so far have been limited to the mean-field case and dimensions $d= 1, 2$, the model's rigorous analysis has been based on expansions and  inequalities.   
 In particular the latter, being less limited by issues of convergence,   allowed to establish the following basic  results on the model's phase transition and critical behavior.

\begin{theorem}[\cite{ABF87}] \label{thm:sharpness} 
For any ferromagnetic Ising model on $\Z^d$, at $d>1$, with a shift invariant  irreducible  pair interaction of exponential decay in  $\|x-y\|$ (the Euclidean distance),
there exists $\beta_c < \infty$ such that: 
\begin{eqnarray}  \label{eq:offcriticality} 
\langle \sigma_x \sigma_y \rangle_{\beta,0}  \,\, \,  
\begin{cases} 
 \,\, \leq A(\beta) \, e^{-\|x-y\| /\xi(\beta)}  & \forall \beta < \beta_c 
 \quad \mbox{with $\xi(\beta) >0, \, A(\beta) < \infty$}
 \\[2ex]  
 \,\,  \geq M(\beta)^2  & \forall \beta > \beta_c\quad \mbox{with $M(\beta) >0$}
 \end{cases}
 \end{eqnarray}  
 
Furthermore,  approaching the critical point $(\beta_c,0)$ from three different directions:
 \be \label{power_laws}
{ M(\beta,0)   \geq     \rm{C} |\beta -\beta_c|_+^{1/2} \, , \quad   
M(\beta_c,h)  \geq   \rm{C} \, h^{1/3} \,, \quad
\frac{\partial M}{\partial h}(\beta,0)  \geq   \frac{\rm{C}}{|\beta_c-\beta|_+^{1}}  }
\ee
with constants $\rm{C}$ which may vary from one equation to another.
\end{theorem} 
Bundled here is a number of statements  of which even the assertion that the low and high temperature regimes extend up to a single transition point (i.e. without an intermediate phase where neither option holds) is beyond the reach of perturbation expansions.  Proofs have been derived through non-perturbative means, mostly based on correlation inequalities and relations expressed in terms of  partial differential inequalities which they enable.

Though stated initially for $\Z^d$, the familiar proofs readily generalize to   Ising spin systems on amenable transitive  graphs with homogenous ferromagnetic interaction of exponential decay.  And with suitable adjustments in  the decay rate in \eqref{eq:offcriticality}, the assumption  of exponential decay of $J_{x,y}$  can be adjusted to slower decay rates, as long as the two-body interaction is summable, i.e.  $\|J\|_1<\infty$.

More can be said for the critical behavior of the nearest neighbor interaction in different dimensions.  And not only about the special case of $d=2$, but also  for higher dimensions.  A significant change occurs as $d$ crosses its ``upper-critical'' value $d=4$, cf. \cite{Aiz82, AizFer86,AizDumSid15,AizDum21} and references therein.  Fundamental role for many of these was played by methods which emerged from the model's random current representation, which is our main subject here.

\section{The random current representation}\label{sec:RCR}

\subsection{RCR's construction} \mbox{ } 

Following is  the RCR definition in the context of finite systems, as it was presented in \cite{Aiz82}.  Let us however note that RCR's  infinite volume limits admit also a direct formulation on the limiting graphs.  As was learned in \cite{AizDumSid15}, the infinite volume RCR formalism allows to combine RCR tools with powerful considerations enabled by translation invariance. 

\begin{definition} 
For an Ising model on a finite graph $\G$  
of vertex and edge sets  $\mathcal V$ and $\mathcal E$, 
with the Hamiltonian 
\be \label{H_no_h}
 {H}_\G(\sigma) := - \sum_{\{x,y\} \in \E(\G)} J_{x,y} \sigma_x \sigma_y 
\,. \ee
a {\em current} configuration  is an integer-valued function defined over the  unordered edges, $\n: \mathcal E  \to \Z_+$.    
The value of $n( (x,y))$ referred to as the current's flux over the edge $(x,y)\in \mathcal E$. 

The collection of vertices of odd incident flux 
\be \label{d_n}
 { \partial\n \, := \, \{ x\in \V :\, (-1)^{\sum_{(x,y) \in \E} n( (x,y))} = -1 \} }
\ee  
is referred to as the current's   \emph{source set}.

At given $\beta \geq 0$, the current configurations are assigned the \emph{weight} function
\be \label{w(n)}
 { w(\n)  :=\prod_{\{x,y\}\subset \E}\frac{\displaystyle(\beta J_{x,y})^{n((x,y))}}{n((x,y))!} } \,.
\ee
\end{definition} 

In these terms, the  Ising model's partition function (at the free boundary conditions) is: 
\begin{equation}\label{eq:8}
Z_{\Lambda,\beta}(J,h) := \frac{1}{2^{|\Lambda|} }\sum_{\sigma: \V\to \{-1,1\}} \prod_{(x,y)\in \E} \exp(\beta J_{x,y}\sigma_x\sigma_y) 
=
  \sum_{\substack{\n:\,\E\to \Z_+ \\ \partial\n=\emptyset}} w(\n) \,.
\end{equation}
This representation emerges through the expansion of  the Gibbs factors 
into its  Taylor  series in $\{J_{x,y}\}$ (which is convergent for all  $\beta \in \C$) 
\be \label{RCR_spin_product}
\exp(\beta J_{x,y}\sigma_x\sigma_y)=\sum_{n( (x,y))\ge0} \frac{(\beta J_{x,y}\sigma_x\sigma_y)^{n( (x,y))}}{n( (x,y))!} \, 
\ee
followed up by the independent average over the spin values.

Correspondingly, the Gibbs state multi-spin correlation functions take the form
\begin{equation}\label{eq:erg}
S_n(x_1, ... , x_n) := { \langle\prod_{x\in A}\sigma_x\rangle_{\Lambda,\beta}=\frac{\displaystyle\sum_{\n:\, \partial\n=A}w(\n)}{\displaystyle\sum_{\n:\partial\n=\emptyset}w(\n)}\,  }
\end{equation}
In this reformulation of the expectation values, Ising spins appear as source-insertion operators acting on the functional integral/sum \eqref{eq:8}.  

It is instructive to note that any current configuration $\n$ is decomposable (in a non-unique way) into the sum of functions  counting the number of times the graph's edges are traversed by a union of a family of potentially overlapping   of  loops with few lines of non-overlapping end points (whose collection forms $\partial \n$).  In the model's infinite volume limit this may also include lines running off  to infinity (at one end or both).

In a sense,  $\n$ resembles the flux of a charge-conserving current, albeit of a charge which is conserved only \emph{mod} 2 (or, alternatively stated:  a parity current).

\subsection{The switching lemma}\mbox{  } 
  
  A feature which makes RCR particularly effective is 
the  combinatorial symmetry which is expressed in the so-called switching lemma.  The switching argument appeared first in the graph theoretic derivation of the Griffiths-Hurst-Sherman inequality~\cite{GHS70}.   
The more comprehensive random current perspective was developed  in \cite{Aiz82} with  results which 
laid down a stochastic geometric picture of the propagation of Ising model's correlations.   These include:
\begin{itemize} 
\item  the spin correlations'  representation in terms of RCR backbone paths 
\item  an RCR-based percolation  characterization of the onset of the Ising model's phase transition 
\item  key features of the  multi-spin spin correlations at the critical state (Gaussian in high dimensions, definitely not in $d=2$, and likely not in $d=3$)
\end{itemize} 
And the list of RCR enabled insights and results still continues to grow.

As we try to strike a balance between the statement's usefulness and simplicity of formulation, following is a version of 
the switching lemma (a term coined in \cite{Aiz82}) 
formulated for pairs of graphs of which one is a subgraph of the other.  The statement is a bit more general than some of its published versions, but the proofs in essence coincide, and rest on a combinatorial argument~\footnote{The switching symmetry is linked with the infinite divisibility of the Poisson distribution 
whose presence can be recognized in the weights $w(\n)$ defined in \eqref{w(n)}.}
which goes back to Griffiths-Hurst-Sherman \cite{GHS70}.

\begin{lemma}\label{lem:switching}  Let  $\G_j = (\mathcal V_j, \mathcal E_j)$,  $j=1,2$,  be a  pair  of graphs, of non-empty intersection $\widehat {\G} = \G_1 \cap \G_2 $ with  two sets of coupling constants  associated with the graphs' edges, $\mathcal J_j: \mathcal E_j \to \R$, which coincide on the common edge set $\widehat {\mathcal E} = \mathcal E_1 \cap \mathcal E_2$.

 Then for any triple of vertex sets $A_j\subset \V_j$ ($j=1,2$), $  B\subset \widehat{\mathcal V}$
\begin{multline} \label{eq:switching}
\sum_{\substack{
\n_1:\partial\n_1=A_1\\
\n_2:\partial\n_2=A_2}}  w(\n_1)w(\n_2) \,   F(\n_1+\n_2)\,  \,  \id_{_{B}} [\n_1+\n_2]  \ = \\   
\sum_{\substack{
\n_1:\partial\n_1=A_1\Delta B\\
\n_2:\partial\n_2=A_2\Delta B }}  w(\n_1)w(\n_2) \,   F(\n_1+\n_2) \,  \,  \id_{_{B}} [\n_1+\n_2]     
\end{multline}
where $\n_j$ are summed over the random current configurations on the corresponding graphs $\G_j$, 
  $\Delta $  denotes the symmetric difference of sets  and 
$\id_{_{B}} [\m]$ is the  indicator function
\be \label{eq:xi}
\id_{_{B}} [\m] = 
\begin{cases}  
1 & {\mbox { $\exists \mathbf {k}: \widehat {\mathcal E} \to \Z_+$, with 
            $\mathbf {k}  \leq \m$ 
            and $\partial \mathbf {k} = B$ }}
\\ 
0& \mbox {otherwise} 
\end{cases} 
\ee 
which in \eqref{eq:switching} is evaluated over the restriction of $\n_1+\n_2$ to  $\widehat {\mathcal E}$.  
\end{lemma}

\begin{proof} 
The claimed relation holds for each specified set of values of the sum $\m=\n_1+\n_2$.  
For an especially simple proof, consider the multigraph representation of the pair of graphs, in which the edges of $\G_1\cup \G_2$ are drawn with multiplicity $\m$.   Let $\mathcal G(\m)$   be the thus constructed multigraph over the vertex set of $\G_1\cup \G_2$.  A useful combinatorial symmetry is revealed by the identity   
\be \label{eq_combs}
w(\n_1) \, w(\n_2)  = w(\m)  
\prod_{\{x,y\} }\left( {\begin{array}{c}
   m_{\{x,y\} } \\
   {n_1}_{\{x,y\}}   
  \end{array} } \right)  \, =:  w(\m)  
 \left( {\begin{array}{c}
   \m  \\
    \n_1  
  \end{array} } \right) \, .
\ee
This allows to rewrite the sums on left  side of \eqref{eq:switching}  in the form
\be  \label{switching_LHS}
\begin{array}{c}
\mbox{LHS} \\[1ex]   
\mbox{RHS} 
\end{array} 
 \mbox{of \eqref{eq:switching}} = 
\sum_{\substack{ \m: \mathcal E_1 \sqcup \mathcal E_2 \rightarrow \Z_+\\
\partial\m=A_1\Delta A_2}}
 w(\m) \,   F(\m)\,  \,  \id_{_{B}} [\m]  \,  \times 
 \begin{array}{c}
\mathcal N_{LHS}  \\[1ex]  
\mathcal N_{RHS}  
\end{array}  
 \ee 
 and 
with $\mathcal N_{LHS}$ and $\mathcal N_{RHS}$
the numbers of decompositions of the multigraph $\G(\m)$ into $\G(\n_1)\sqcup \G(\n_2)$ 
with $\G(\n_j)$ meeting the 
support and source conditions of the  sums on the left and the right sides of \eqref{eq:switching}, correspondingly.    

Under the condition imposed by $\xi_B[\m]$ 
the set of edges of $\widehat \G(\m)$ includes at least one set of edges whose source set is $B$.  
It is important, and easy to see (e.g.~by the set's countability) that among those a  choice of one such 
subgraph (whose edge count $k$ mees the conditions spelled in \eqref{eq:xi}) can be made as a  measurable functional of $
\G(\m)$.  Given such a choice, the invertible replacement in the last sum of the pair $\G(\n_1)$ and $\G(\n_2)$  by their symmetric difference with that fixed subgraph  shows that for each $\m$ (with $\id_{_{B}} [\m] \neq 0$) 
\be 
\mathcal N_{LHS} = \mathcal N_{RHS} 
\ee 
Hence the two terms in \eqref{eq:switching} are equal, as claimed.
\end{proof}  

\subsection{Emergence of a stochastic geometric picture}  \mbox{}

Among the immediate implications are the following exact relations, that hold for the Ising model over any finite graph (which is sometimes omitted in the abbreviated notation).  
 \begin{eqnarray}  \label{pf:G2}
\frac{ \langle \sigma_A  \rangle  \, \langle \sigma_B \rangle} 
{\langle \sigma_A  \, \sigma_B \rangle }   &=&   
 \sum_{\substack{
\partial\n_1=A \Delta B  \notag \\
\partial\n_2=\emptyset}} 
 w(\n_1) w(\n_2) \,  \,  \id_{_{A}} [\n_1+\n_2]  \big / 
{ \sum_{\substack{
\partial\n_1=A \Delta B \\
\partial\n_2=\emptyset}} 
 w(\n_1) w(\n_2) 
}  \\[2ex]
  &=&  \mathbb {E}^{A \Delta B , \emptyset }\big( \id_{_{A}} [\n_1+\n_2]  \big)  \leq 1
\end{eqnarray}  
where $\mathbb {E}^{A \Delta B , \emptyset }$ denotes the \emph{probability}  expectation value over a pair of random currents  with the indicated source sets,   
$\n_1+\n_2$ supports a subgraph of sources $A$.

In particular, for any pair of sites $x,y$: 
\be \label{eq:XtoY}
\langle \sigma_x \sigma_y \rangle^2 =    \mathbb {P}^{\emptyset, \emptyset } \big( x 
\stackrel{\n_1+\n_2}{   \longleftrightarrow}  y\big) \,.
\ee 
 This readily implies that in the natural infinite volume limit the onset of \emph{long range order} 
 ($\langle \sigma_x \sigma_y \rangle _\beta > M(\beta) >0$) coincides with the \emph{percolation} transition in a duplicated system of two independent sourceless random currents.  
 
The structure of the higher order correlation functions is expressed in terms of  the Ursell functions $U_n$, of which of the bellwether is 
 \be \label{def_U4}
U_4(x_1,...,x_4) =S_4(x_1,...,x_4)   - 
\left[ \langle 1, 2 \rangle  \langle 3, 4 \rangle + \langle 1, 3 \rangle  \langle 2, 4 \rangle  
+ \langle 1, 4  \rangle  \langle 2, 3 \rangle  \right]
\ee
where $\langle j, k \rangle = S_2(x_j,x_k)$.  If $|U_4(x_1,...,x_4)|/s_4(x_1,...,x_4) \ll 1$, then the correlations are close to Gaussian, at least at the level of $n=4$ (further inequalities imply that from $n=4$  follows the rest~\cite{New75, Aiz82}).  

A simple applications of the switching lemma \eqref{pf:G2} yields: 
\begin{theorem}[\cite{GHS70,Aiz82}]  For any  triplet of sites:
 \be \label{xyz}
\frac{\langle \sigma_x \sigma_y \rangle  \, \langle \sigma_y \sigma_z \rangle}     
{\langle \sigma_x \sigma_z \rangle } \ = \ 
   \mathbb {P}^{\{x,z\}, \emptyset } \big( x 
\stackrel{\n_1+\n_2}{   \longleftrightarrow}  y\big) \,.
\ee
And for any quadruplet of sites  
\begin{subequations} 
\begin{align}
\label{U4a}   
\frac{-1}{2}\, U_4(x_1,...,x_4)  
&=&  \!\!\!S_2(x_1,x_2)  \,  S_2(x_3, x_4) \,\, 
\mathbb {P}^{\{x_1,x_2\}, \{x_3,x_4\} } 
\big(x_1  
\stackrel{\n_1+\n_2}{\longleftrightarrow}  x_3 \big) 
 \\   
&=&   S_4(x_1,...,x_4)  \,  \, \, 
\mathbb {P}^{\{x_1,...,x_4\}, \emptyset } 
\big(  \forall j,k:  \, x_j  
\stackrel{\n_1+\n_2}{\longleftrightarrow}  x_k\, \big)  
\end{align}  
\end{subequations}
\end{theorem} 

Reading \eqref{xyz} it should be noted that under the indicated source conditions 
within the multigraph $\G(\n_1+\n_2)$  
there exists at least one (typically many more)   simple path from $x$ to $z$.  The above ratio of the correlation functions is identified here with the probability that there is such a path which on its way passes through $y$.    

In \eqref{U4a}   the  source conditions imply that 
$x_1  \stackrel{\n_1+\n_2}{   \longleftrightarrow}  x_2$ and $x_3  
\stackrel{\n_1+\n_2}{   \longleftrightarrow}  x_4$.  The Ursell function is then the conditional probability that 
the paths pairing the sources are inter connected in $\G(\n_1+\n_2)$.

These observations, give rise to a stochastic geometric picture of  the propagation of spin correlations through source-pairing paths superposed over, or rater mixed in, a random set of  potentially overlapping  loops.

\section{RCR based path expansions} \label{sec:path}

To sharpen the above observation about the conservation of parity into a working tool, in \cite{Aiz82,AizGra83,Aiz85}
 it was found useful to extract from the RCR expression for  correlation functions  
 a ``backbone'' of paths supported on edges of odd flux,  leaving the rest as a random sourceless flux configuration.
 
The  decomposition of the random current's flux into an odd ``backbone'' and loops is not unique.   
In \cite{Aiz82} the backbone paths were extracted by starting from a randomly selected source,  then proceeding through one of the 
 edges of odd 
$\n_b$ which were not yet traversed, selected at random, and continuing in this fashion till another source is reached.    If by that time not all sources were paired,  the next  one will be selected at random and the construction continued.  The construction stops once the source set is fully paired.

In subsequent works it was found more convenient to base the path selection on a predetermined ranking of the sites and (separately) edges of the finite graph.  We retain here  this choice.    

In this manner one gets  a representation of the form
\be 
\langle \sigma_{x_1} ... \sigma_{x_{2n}} \rangle \ = \  \sum_{\mbox{pairings}} \rho(\gamma_1, \gamma_2,..., \gamma_n)
\ee
where $\{\gamma_j\}_{j=1}^n$ are paths by which the $2n$ sources are paired, each path expressed as a sequence of concatenated edges $b$.  (To avoid confusion with Euler's constant $e$, we use  $b$ as the generic symbol for edges.)

The weights $\rho(\gamma_1, \gamma_2,..., \gamma_n)$ are obtained by summing the Gibbs  weights $w(\n)$ of random currents which yield the  given collection of backbone paths (under the preselected ranking of sites and edges).     

In any such deterministic algorithm the specificaiton of $\gamma$ implies not only that $\n$ is odd along the edges covered by $\gamma$, but also that its values are even at every edge which should have been been included in $\gamma$ had that not been the case.  We denote by $\widehat \gamma$ the collection of edges of both type.   

It  follows~(cf. \cite{Aiz85}) that under the deterministic algorithm the weights, for sequences of paths in $\Lambda$, are of the form
\be
  \rho(\gamma_1, \gamma_2,..., \gamma_n) \ = \  I[\gamma_1, \gamma_2,..., \gamma_n] \ \zeta_\Lambda(\gamma_1, \gamma_2,..., \gamma_n) 
  \ \prod_{j=1}^n \prod_{b\in \gamma_j} \tanh(\beta J_b) 
 \ee 
where $ I[\gamma_1, \gamma_2,..., \gamma_n] $ is a binary indicator function which is $1$ if the paths are consistent with the constraints which are naturally associated with the path selection algorithm, and    
\be \label{eq_zeta}
\zeta(\gamma_1, \gamma_2,..., \gamma_n) \ = \  \frac{Z'_\Lambda(\widehat \gamma_1, \widehat\gamma_2,..., \widehat \gamma_n)}{Z_\Lambda} \  \prod_{j=1}^n \prod_{b\in \widehat \gamma_j} \cosh(\beta J_b) 
\ee 
where  $Z'_\Lambda(\widehat \gamma_1, ..., \widehat \gamma_n)$ is the partition function under a depleted interaction, obtained by setting  $J_b=0$ for all the edges $b\in \cup_{j=1}^n \widehat \gamma_j$.

\begin{thm}  For the random walk expansion based on the deterministic algorithm for the selection of RCR backbone paths the following holds  regardless of the order of sites and edges which guides the path selection: 
\begin{enumerate}  \label{thm:path_exp}
\item[i)] Denoting by $ \gamma_1\circ \gamma_2$ the concatenation of  paths, in case the second starts from the endpoint of the first, one has 
\be \label{RW_i}
\rho(\gamma_1\circ \gamma_2)  \ = \  \rho(\gamma_1, \gamma_2)
\ee
(which is order dependent, i.e. need not equal $\rho(\gamma_2, \gamma_1)$) 
\item[ii)] The partial sums of $\rho(\gamma_1,...,\gamma_n)$ at specified sequence of end points $\{x_1,...,x_{2n}\}$
satisfy
\be \label{RW_ii}
\sum_{\gamma_n: x_{2n-1} \mapsto x_{2n}} \rho(\gamma_1, \gamma_2,...,\gamma_n)   \ = \ \rho(\gamma_1,...,\gamma_{n-1}) \  
\langle \sigma_{x_{2n-1}} \sigma_{x_2} \rangle_{( \cup_{j=1}^{n-1} \widehat \gamma_j)^c} 
\ee 
where
the subscript indicates that the spin-spin correlation is to be computed for the weakened interaction, obtained by setting $J$ to $0$ on the edges of 
$\{\widehat \gamma_1, ..., \widehat \gamma_{n-1}\}$. 

\item[iii)] For any path in $\Lambda$,  $\zeta_\lambda(\gamma) \leq 1$, and at fixed $\gamma$ the  function is decreasing in $\Lambda$.
\item[iv)]  The weight $\zeta_\Lambda(\gamma_1, \gamma_2,..., \gamma_n) $ is ``super-multiplicative'' under splits of the sequence.  In particular 
\be \label{RW_iv}
\zeta_\Lambda(\gamma_1, \gamma_2)  \geq \zeta_\Lambda(\gamma_1) \ \zeta_\Lambda(\gamma_2)
\ee   
\end{enumerate}

\end{thm} 

Properties $i)$ and $ii)$ follow readily from the path definition, $iii)$ and $iv)$ are consequences of the Griffiths first and second inequality.   In essence:  aside from  the strictly multiplicative factor in \eqref{eq_zeta} the ratio $Z'/Z$ which appears there represents the multiplicative effect of weakening the interaction.  As it was put in \cite{Aiz85} (eq. (3.17) there),  Griffiths second inequality implies that for the ferromagnetic spin models it is easier to weaken an already weakened interaction.  

It is useful to note  that, as a consequence of the above, the path weights satisfy  the dichotomy:
\be \label{eq_alt} 
\rho(\gamma_1,\gamma_2)    
\begin{cases}  
\geq \rho(\gamma_1) \rho(\gamma_2) & \mbox{if the paths  form a compatible sequence} \\[1ex]  
0 & \mbox{otherwise} 
\end{cases} 
\ee

\noindent Two examples of these tools utility: 
 
1)  The backbone expansion and the dichotomy \eqref{eq_alt} allowed to establish that
 the Ursell function
satisfies the tree-diagram  bound 
\be \label{tree_diag}
|U_4(x_1,x_2,x_3,x_4)|  \, \leq \,  2\,   \sum_{u\in \mathcal V} \,  \prod_{j=1}^4\, 
 \langle \sigma_{x_j} \sigma_{u} \rangle  \,.   
\ee
Combined with the information obtained from the infrared Gaussian-domination bound of \cite{FILS78}, 
\eqref{tree_diag} 
played a key role in the proof that the Ising model's scaling limits are Gaussian in dimensions $d>4$ \cite{Aiz82}  (Part II), and the more recent extension of the result to $d=4$ \cite{AizDum21}.  

2)  The first two properties, $i)$ and $ii)$, enable optional stopping arguments which are natural for path expansion.  Their utility is demonstrated  below in an  alternative proof of  Simon-Lieb inequalities (Theorem~\ref{SL_ineq}).

   \subsection{A constructive exchange} \mbox{} 
   
The partial similarity between  the stochastic geometry underlying the Ising correlation and percolation has led to a  constructive dialog.  RCR-based relations bounds on the multi spin correlation functions have inspired analogous, yet somewhat different, bounds for the connectivity functions of independent percolation models, reachable by other methods \cite{Aiz82,AizNew84}, and related but different criteria for the upper-critical dimension.   That was followed by proofs of  sharpness of the phase transitions and mean-field  critical exponent bounds, but here progress proceeded in the  reverse order  \cite{AizBar87, ABF87}.  
Though only partial, the overlap in the results raised the question whether the two models can be handled simultaneously.  In time, this challenge was   reached (though still partially) through yet another technique~\cite{DumTas15}  which is applicable to both, and beyond - to the broader class of FK models that is presented next.

\section{The Fortuin Kasteleyn random cluster representation}  \label{sec:FK}

 \subsection{The FK representation} \mbox{}

As a tool for a stochastic geometric perspective on the Ising models's Gibbs states 
the \emph{random current} representation competes for our attention with the 
 Fortuin-Kasteleyn (FK) \emph{random cluster} representation~\cite{ForKas72,  Gri06}.      
 Each  reveals a stochastic geometric underpinning of  the spin correlations, though in a somewhat different manner.   
Before commenting further on the relation between the two, which is of relevance 
 for our discussion in Section~\ref{sec:deconf} let us briefly recall the Fortuin-Kasteleyn construction.

The FK representation starts from the identity
\be \label{FK_standard}
e^{\beta J_{u,v}  \sigma_u \sigma_v } = e^{\beta J_{u,v}} \left[ p_{u,v} \cdot  \delta_{\sigma_u\sigma_v, 1} + (1-  p_{u,v} ) \cdot 1 \right]
\ee
with $ p_{u,v}  = 1- e^{-2\beta J_{u,v}}$. 
One may read in this a polarization of the Gibbs factor associated to the edge $(x,y) \in \mathcal E$
into a linear superposition of two term, under the one with $\delta_{\sigma_u\sigma_v, 1} $ the spins are constrained to agree and under the other they do not interact at all:

For the  Ising model's partition function 
under the two-body Hamiltonian \eqref{H_no_h} this yields  
\begin{multline}\label{eq:18}
Z_{\Lambda,\beta}(J)  :=  \frac{1}{2^{|\Lambda|} } \sum_{\sigma: \V\to \{-1,1\}} 
\prod_{(x,y)\in \E} e^{\beta J_{x,y}\sigma_x\sigma_y) }  =
  \\   
  =  \frac{e^{\beta \sum_{b\in \mathcal E} J_b}}{2^{|\Lambda|} }  
  \sum_{\substack{{\sigma: \V\to \{-1,1\}} \\ \omega \subset \mathcal E}}\,\, 
\prod_{b\in \omega}  p_b 
\prod_{b\in \mathcal E \setminus \omega} (1-p_b)  \cdot 
\id_\omega \left[\sigma \right]  
\end{multline}
where $b$ and $(x,y)$ are used interchangeably to denote the graph's edges,  and the last factor  is the indicator function
\be\notag
\id_\omega \left[\sigma \right] = \begin{cases} 
1 & \mbox{if $\sigma_x$ is constant on each of the $\omega$-connected clusters} \\ 
0 & \mbox{otherwise} 
\end{cases}
\ee

Summing over the spin degrees of freedom one gets:
\be\label{eq:28}
Z_{\Lambda,\beta}(J,h)  
=  \frac{e^{\beta \sum_{b\in \mathcal E} J_b}}{2^{|\Lambda|} }  
  \sum_{ \omega \subset \mathcal E}\,\, 
\prod_{b\in \omega}  p_b 
\prod_{b\in \mathcal E \setminus \omega} (1-p_b)  \cdot 
q^{N(\om)}
\ee  
at $q=2$, where $N(\om)$ is the number of $\m$-connected clusters in the corresponding decomposition of $\mathcal V$, and  $q$ is the number of states of each spin.  

Associated with any decomposition of the partition function into a sum of positive terms is a probability distribution of the relevant degrees of freedom, that is obtained by normalizing the summed weights by the factor $1/Z_{\Lambda,\beta}$.    The FK joint probability distribution of the spin and edge variables $\{ \sigma_x\}_{x\in \mathcal V}$ and $\{ \omega _b\}_{b\in \mathcal E}$, 
 can be read in this manner from \eqref{eq:18},  and the marginal distribution of just the edges can be read from \eqref{eq:28}.
 The latter forms the Fortuin-Kasteleyn random cluster model.  

A  notable features  of the joint distribution of  $\{ \sigma_x\}_{x\in \mathcal V}$ and $\om=\{ \omega _b\}_{b\in \mathcal E}$    
is the simplicity of  the conditional distribution of each of these two sets of variables conditioned on the other:   

i) Conditioned on $\om$, the spins $\sigma_x$ are constrained to ``vote'' as cliques: full agreement within each $\om$-connected cluster, no correlation between different clusters. 

ii) Conditioned on the spin configuration, the edge variables $\om$ are distributed independently of each other, subject only to the constraint that $n((x,y)) =0$ 
where  $\sigma_x\neq \sigma_y$.

Partial summation over $\{ \sigma_x\}$, and normalization of the terms in the resulting \eqref{eq:28},  leads to a  joint distribution of the edge variable $\om$ which is known as the ``$q$-random cluster model'', in this case at $q=2$ (\cite{ForKas72,  Gri06}).

From the above description of the conditional distribution of the spins under $\om$ one readily concludes that for any pair of sites $x,y \in \mathcal V$:  
\be \label{FK:XtoY}
\langle \sigma_x \sigma_y \rangle =    \mathbb {P}^{\text{FK}(2)}_{\Lambda, \beta} \big( x 
\stackrel{\omega}{   \longleftrightarrow}  y\big) \,.
\ee 
with $\mathbb {P}^{FK} $ the marginal probability distribution of the edge variables $\om$ corresponding to \eqref{eq:28}.  

A comparison of   \eqref{FK:XtoY} and  RCR's \eqref{eq:XtoY} leads to the  striking observation that 
the Ising model's onset of \emph{long range order} coincides with a pair of  symultaneous with both 
the  double random current's percolation transition, and the percolation transition in the corresponding FK random cluster model. 
However along with that one may also notice  a curious difference:   the point-to-point connection probabilities for the FK clusters and those 
for the doubled RCR are not the same.  The relation between the two (linked through  $\langle \sigma_x \sigma_y \rangle_{\Lambda, \beta}$) is: 
\be
\mathbb {P}^{\text{FK}(2)}_{\Lambda, \beta} \big( x \stackrel{\omega}{   \longleftrightarrow}  y\big) = 
\sqrt{\mathbb {P}^{\emptyset, \emptyset}_{\Lambda, \beta} \big( x \stackrel{\n_1+\n_2}{   \longleftrightarrow}  y\big)} \, , 
\ee 

\subsection{Comparison of the methods}  \mbox{}  

The FK random cluster representation places the Ising model within a  family of interactive percolation models indexed by the continuous parameter $q$.  Among those are the independent percolation ($q=0$), Ising ($q=2$), and other $q$-state spin models ($q\in \mathbb N$).
The critical behavior in these models is not $q$-independent, however certain arguments apply for a broader range of this parameter than just $q=2$.  

Among the useful tools related to  the FK representation one finds:
\begin{itemize}
 \item The Swendsen-Wang  algorithm~\cite{SweWan87} (which Edwards-Sokal~\cite{EdwSok88} extended to  other systems)  for numerical simulation of the Gibbs state though 
a sequence of projections onto corresponding conditional distributions, conditioned on $\om$ and $\sigma$, alternatively.

\item The Fortuin-Kasteleyn-Ginibre (FKG) inequality~\cite{FKG71}, which implies  that monotone increasing functions of $\om$ are positively correlated, in the sense that under the above probability any pair of functions $F,G: \mathcal E \to \R$ which are coordinate-wise monotone increasing  are positively correlated: 
\be
\mbox{$F\nearrow$ and $G \nearrow$}  \Rightarrow \quad 
\big\langle F G \big\rangle^{\text{FK}(2)}_{\Lambda, \beta} \geq   \big \langle F \big\rangle^{\text{FK}(2)}_{\Lambda, \beta}  \, \big\langle G \big\rangle^{\text{FK}(2)}_{\Lambda, \beta} 
\ee

\item Variance bounds based on  algorithmic exploration (cf. \cite{OSSS05, DRT17}. 

\end{itemize} 

To which one may add that the different FK(q) probability distributions can be related  through the  FKG  inequalities, which  allow to  compensate for the difference in $q$ through suitable adjustments of $\beta$~\cite{FKG71, ACCN88}.  
(The last reference includes a notable result for which this relation was put to use).

Compared to the FK \emph{random cluster} representation the more specifically Ising-model's RCR \emph{random current} representation offers  
a sharper tool for  features  in which the parity conservation  plays a role (in consequence of the Ising model's spin flip symmetry).  Related to that is the difference between the Ising model's $\phi^4$ rather than percolations's $\phi^3$ type diagrammatics, which are linked to   the differences in the critical exponents and the upper critical dimensions ($d=4$ for Ising and $d=6$ for percolation).   

Thus, the RCR method within its narrower range offers some advantage. 
 The difference is amplified by the availability of  reflection positivity arguments for Ising models~\cite{FILS78}, the likes of which have not been found for percolation and non-integer $q>1$ valued FK(q) models.   
However,  when both techniques are applicable, which is not always the case, the arguments are sometimes simpler in their  FK version.    An example of that can be found below, in Section~\ref{sec:deconf}.   

 \section{Geometric representations' adjustment to frustration}  \label{sec:frustration} 

\subsection{Frustration}  \mbox{}

Frustration is said to be present when not all the term in the Hamiltonian can be minimized simultaneously.  
In the case of  
\be\label{H_only_J}
H = - \sum_{\{x,y\}\in \mathcal E} J_{x,y} \sigma_x \sigma_y 
\ee 
with the free  or locally  wired boundary conditions,    
frustration is present  if for some loop of edges the product of the coupling coefficients is negative.   
 
 Frustration  may also result  from the inclusion of  external fields or  boundary conditions of mixed signs.   For a unified treatment,  the latter two cases can be turned into \eqref{H_only_J}
through the Griffiths's ghost spin trick (recalled below). 

\subsection{Frustration-adjusted FK measures}  \mbox{} 

The FK representation's  adaptation to frustration bearing Hamiltonians was presented, and discussed in the context of spin glass models, by  Y Kasai and A. Okiji~\cite{KasOki88}, and Gandolfi-Kean-Newman~\cite{GKN92}[Ch. 4]. 
For such an extension,  in the case of antiferromagnetic (negative) couplings the polarizing decomposition \eqref{FK_standard}  is  changed to read:
\be \label{FK_mixed}
e^{\beta J_{u,v}  \sigma_u \sigma_v } = e^{\beta |J_{u,v}|} \left[ p_{u,v} \cdot  \delta_{\sigma_u\sigma_v, \, \text{sgn}(J_{u,v})} + (1-  p_{u,v} ) \cdot 1 \right]
\ee
with $ p_{u,v}  = 1- e^{-2\beta |J_{u,v}|}$.  

Under this convention  
 $p_{u,v} \in [0,1]$  regardless of the coupling's sign.   
The sign of $J_{u,v}$  affects the constraint in $\delta$, which in the \emph{antiferromagnetic} case requires the spins to \emph{disagree}.

The corresponding extension of \eqref{eq:18} and \eqref{eq:28} under the  coupling ${\bf J} =\{J_b\}_{b\in \mathcal E}$ is
\begin{multline}\label{eq:18mod}
Z_{\Lambda,\beta}(J,h) =    
   \frac{e^{\beta \sum_{b\in \mathcal E} |J_b|}}{2^{|\Lambda|} }  
  \sum_{\substack{{\sigma: \V\to \{-1,1\}} \\ \omega \subset \mathcal E}}\,\, 
\prod_{b\in \E}  p_b 
\prod_{b\in \E \setminus \omega} (1-p_b) \cdot 
\id^{\bf J}_\omega \left[\sigma \right]   
\end{multline}
with the {\bf J}-adapted  constraint: 
\be\label{constraint}
\id^{\bf J}_\omega \left[\sigma \right] = \begin{cases} 
1 & J_{u,v} \sigma_u \sigma_v \geq 0 \quad \mbox{for all  $(u,v)$ with $\omega_{(u,v)} =1$} \\ 
0 & \mbox{otherwise} 
\end{cases} \,. 
\ee

For a given edge configuration $\omega \subset \mathcal E$  sites $u$ and $v$ are said to be $\omega$-connected if there is a path linking the two along edges of $\omega$.    

\begin{definition}  \label{def:FF_FK}
For a system with a  Hamiltonian of the form \eqref{H_only_J} an edge configuration $\omega\subset \mathcal E$ is said to be {\bf J}-\emph{frustration-free} (FF) if  
closed loops supported on $\omega$ include only even numbers of edges with $J_b<0$. \\  
The indicator function of this condition is denoted here $\id^{\bf J}_{FF}[\omega]$.    
\end{definition}

This terminology allows to define the following $\{-1,0,1\}$-valued relative parity function  
\be  \label{sgn_omega}
 \sgn_{\mathcal J}(u, v ; \omega) = 
 \begin{cases}  
 \mbox{the parity  of $\omega$-connection}  & \mbox{if $\omega$ is FF and $u \stackrel{\omega}{\longleftrightarrow} v $}\\[1ex]  
     0   & \mbox{otherwise}    
     \end{cases}   \, 
\ee
where the parity  of $\omega$-connection is  the $\pm1$ parity of the numbers of 
negative edges ($J_b<0$) along paths which $\omega$-connect the specified pair of sites $\{u, v\}$. 
Consistency of this definition for all $u,v\in \mathcal V$ is equivalent to $\omega$ being $\bf J$-frustration free.  

In these terms one has the following generally known extension of the Fortuin-Kasteleyn representation to systems bearing frustration.

 \begin{theorem}[Frustration-adjusted FK]  For finite Ising systems with a Hamiltonian of the form \eqref{H_only_J} (not necessarily ferromagnetic) the partition function and  the corresponding Gibbs correlation function admit the following representation 
 \begin{eqnarray}  \label{FK_mod_Z}
Z_{\Lambda,\beta}(J,h) &=&  
Z_{\Lambda,\beta}(|{\bf J}|,h) \, \cdot \, \mathbb P_{\Lambda,\beta}^{\text{FK}(2)}
\left( \mbox{$\omega$ is $\bf J$-frustration free} 
\right)
\end{eqnarray}  
\begin{eqnarray}\label{FK_mod_corr}
\langle \sigma_u \sigma_v \rangle_{\Lambda,\beta}^{\mathcal J}    
&=&   
\frac
{\mathbb E_{\Lambda,\beta}^{\text{FK}(2)}
\left(\sgn_{\mathcal J}(u, v ; \omega) \right)}
{\mathbb P_{\Lambda,\beta}^{\text{FK}(2)}
\left( \mbox{$\omega$ is $\bf J$-frustration free} 
\right)} \notag \\[2ex]  &=& 
\mathbb E_{\Lambda,\beta}^{\text{FK}(2)}
\left(\sgn_{\mathcal J}(u, v ; \omega) \, \big| \, \mbox{FF} \right).
\end{eqnarray}
where the probabilities and expectation values refer to the FK random cluster measure of the modified ferromagnetic interaction ${\bf |J|} = \{ |J_b|\}_{b\in \mathcal E}$  
\end{theorem} 
\begin{proof} 
Starting from \eqref{eq:18mod} we note that the  constraint \eqref{constraint} restricts the values of spins in a manner which cannot be satisfied unless $\omega$ is frustration free, and if it is then it allows exactly two spin configurations within each $\omega$-connected cluster.  Hence, the sum  over the spins yields
\be
Z_{\Lambda,\beta}(J,h) =  \frac{e^{\beta \sum_{b\in \mathcal E} |J_b|}} {2^{|\Lambda|} }  
  \sum_{ \omega \subset \mathcal E}\,\, 
\prod_{b\in \omega}  p_b \prod_{b\in \mathcal E\setminus \omega} (1-p_b) \, 
2^{N(\om)}  \cdot \id^{\bf J}_{FF}[\omega]   \,.
\ee 
Without the indicator function factor the sum equals the partition function of the modified ferromagnetic interaction.  Factoring that out the sum is naturally recast  in the form of \eqref{FK_mod_Z}.

The spin-spin correlation function is given by a ratio of the sum of partition function terms multiplied by $\sigma_u \sigma_v$, over the pure partition function.   The considerations explained next to Definition~\ref{def:FF_FK}, lead to 
\begin{multline} 
\langle \sigma_u \sigma_v \rangle_{\Lambda,\beta}^{\mathcal J}     Z_{\Lambda,\beta}(J,h) =  \\  
= \frac{e^{\beta \sum_{b\in \mathcal E} |J_b|}} {2^{|\Lambda|} }  
  \sum_{ \omega \subset \mathcal E}\,\, 
\prod_{b\in \omega}  p_b \prod_{b\in \mathcal E\setminus \omega} (1-p_b) \, 
2^{N(\om)}  \cdot \id^{\bf J}_{FF}[\omega]   \cdot  \sgn_{\mathcal J}(u, v ; \omega)  \\  
= 
Z_{\Lambda,\beta}(|{\bf J}|,h) \, \cdot \, \mathbb E_{\Lambda,\beta}^{\text{FK}(2)}
\left( \id^{\bf J}_{FF}[\omega]   \cdot 
\sgn_{\mathcal J}(u, v ; \omega) 
\right)
\, . 
\end{multline} 
Combined with \eqref{FK_mod_Z} this transforms into
\be 
\langle \sigma_u \sigma_v \rangle_{\Lambda,\beta}^{\mathcal J}     = 
\frac
{\mathbb E_{\Lambda,\beta}^{\text{FK}(2)}
\left(\sgn_{\mathcal J}(u, v ; \omega) \,\, \id^{\bf J}_{FF}[\omega]  \right)}
{\mathbb E_{\Lambda,\beta}^{\text{FK}(2)}
\left( \id^{\bf J}_{FF}[\omega] 
\right)}
\ee
The claimed \eqref{FK_mod_corr} is this relation recast in terms of conditional probability. 
\end{proof}

An application of this representation is presented below.  First however let us present the RCR's adaptation to Hamiltonians with frustration.

\subsection{Frustration-adjusted RCR} \mbox{ }  \label{sec:F_RCR}

In the presence of frustration, the partition function's RCR straightforward  expansion yields terms of mixed signs.  Still,  also in that case the random current representation is of relevance.  As a first step: it allows to identify exact cancellations upon which the partition function regains its status as a sum of positive terms.     

In what follows the terms \emph{probability} and \emph{expectation value} are reserved for averages over non-negative weight.   To keep track of the  signs, we denote
\be \label{|w|}
   w^{\mathcal J} (\n) :=  w^{\mathcal {|J|}}(\n)    =\prod_{\{x,y\}\subset \E}\frac{\displaystyle(\beta |J_{x,y}|)^{n( (x,y))n( (x,y))}}{n( (x,y))!}  \quad  ( \geq 0)\,,   
\ee
and write the extension of \eqref{w(n)} to general couplings  as  
\be 
\prod_{\{x,y\}\subset \E}\frac{\displaystyle(\beta J_{x,y})^{n( (x,y))}}{n( (x,y))!}  =  w^{\mathcal J}(\n) \,  (-1)^{\mathcal F ^{\mathcal J} (\n) }   
\ee 
 $F^{\mathcal J} (n)$  being the total flux of $\n$ over ``negative edges'' (i.e. those with  $J_{x,y}<0$).

In this notation, \eqref{eq:8} extends beyond ferromagnetic systems as:
\be\label{eq:80}
Z_{\Lambda,\beta}(J)
=
 \sum_{\substack{\n:\,\E\to \Z_+ \\ \partial\n=\emptyset}} w(\n) \, (-1)^{\mathcal F^{\mathcal J} (\n)  }\,.
\ee

In parallel to Def.~\ref{def:FF_FK}, for RCR we employ the following terminology.

\begin{definition} 
For systems with a Hamiltonian of the form \eqref{H_only_J} a current configuration  
 $\m: \mathcal {E} \to \Z_+$ is said to be \emph{frustration-free} (FF) if 
 closed loops supported on 
 edges with $m_b \neq 0$ include  only even numbers of negative couplings.  
The indicator function of this event is denoted by $\id_{FF}[\m]$.  
\end{definition}
 
In analogy to \eqref{sgn_omega}  we denote
\be  \label{sgn_RCR}
 \sgn_{\mathcal J}(u, v ; \m) = 
 \begin{cases}  
 \mbox{the parity  of $\m$-connection}  & \mbox{if $\m$ is FF and $u \stackrel{\m}{\longleftrightarrow} v $}\\[1ex]  
     0   & \mbox{otherwise}    
     \end{cases}   \,. 
\ee  

Following is a relevant observation.
 
\begin{lemma}   For any set of couplings $\mathcal J$ and an edge function $\m: \mathcal {E} \to \Z_+$
\begin{eqnarray}   \label{F_sign}
(-1)^{\mathcal F^{\mathcal J} (\m) }  \,  \id_{FF}[\m]    
&=& \id_{FF}[\m]  \cdot 
 \begin{cases} 1 & \mbox{if $\partial \m=\emptyset$} \\  
 sgn_J(u, v ; \m) & \mbox{if $\partial \m=\{u,v\} $}  \,.   
 \end{cases} 
\end{eqnarray}  
\end{lemma} 
\begin{proof}  
The flux function  $\m$ can be presented as the sum of contributions of possibly overlapping loops plus  a simple path in the second case.  Under such a decomposition $\mathcal F^J(\m)$ is additive, and for frustration free  currents $\m$ it  gets zero contribution from the loops, and  $ sgn_J(u, v ; \m) $ from the path.  
\end{proof} 

As we learned from the switching lemma, it helps to duplicate. This remains true also for systems with frustration.     
\begin{theorem}[Frustration-adjusted RCR] \label{RSR_FF}
 For any finite Ising system with a Hamiltonian of the form \eqref{H_only_J} (not necessarily ferromagnetic):
\begin{eqnarray}\label{F_dup}
\frac{Z_{\Lambda,\beta}(\mathcal J)}{ Z_{\Lambda,\beta}(\mathcal{|J|})} &=&  \!\!\!\!\!
    \sum_{\substack{\n_1:\,\E\to \Z_+ \\ \partial \n_1=\emptyset}} 
   \sum_{\substack{\n_2:\,\E\to \Z_+ \\ \partial\n_2=\emptyset}} \frac{w(\n_1)}{{ Z_{\Lambda,\beta}(|J|)} } \,  
   \frac{w(\n_2)}{{ Z_{\Lambda,\beta}(|J|)} } 
    \id_{FF}[\n_1+\n_2]      \notag  \\   \\   \notag 
    &=:&  \bbE^{\emptyset, \emptyset}_{\Lambda,\beta} \left( \id_{FF}[\n_1+\n_2] \right)   \equiv  \bbP_{\Lambda,\beta}^{\emptyset, \emptyset} \left( \mbox{$\n_1+\n_2$ is FF} \right)   \,   
\end{eqnarray}
and   
\begin{eqnarray}\label{corr_dup}
\langle \sigma_u \sigma_v \rangle_{\Lambda,\beta}^{\mathcal J}  \cdot 
\langle \sigma_u \sigma_v \rangle_{\Lambda,\beta}^{\mathcal {|J|}}  
=  
\bbE^{\emptyset, \emptyset}_{\Lambda, \beta} \left ( sgn_J(u, v ; \n_1+\n_2) \, \big \vert \, \mbox{$\n_1+\n_2$ is FF} \right) \,, 
\end{eqnarray}
where $ \bbP_{\Lambda,\beta}^{\emptyset, \emptyset} $ and $ \bbE^{\emptyset, \emptyset}_{\Lambda, \beta} $ denote the probability and expectation value  for a pair random currents, of  sources indicated in the subscript, in the state corresponding to the  ferromagnetic coupling  $\bf J = \{|J_{uv}|\}$.  
\end{theorem} 
    Here, and below, the $\bbE(\, -\, | \,Y\,)$ denotes conditional expectation, conditioned on the event $Y$.

\begin{proof}  We start by adding a reference system and expanding: 
\begin{eqnarray}\label{F_dup2}
\frac{Z_{\Lambda,\beta}(J)}{ Z_{\Lambda,\beta}(|J|)} &=&  \frac{Z_{\Lambda,\beta}(J) \cdot Z_{\Lambda,\beta}(|J|) }{ Z_{\Lambda,\beta}(|J|) \cdot Z_{\Lambda,\beta}(|J|)} \notag \\   
&=&
    \sum_{\substack{\n_1:\,\E\to \Z_+ \\ \partial\n_1=\emptyset}} 
   \sum_{\substack{\n_2:\,\E\to \Z_+ \\ \partial\n_2=\emptyset}} \frac{w(\n_1)}{{ Z_{\Lambda,\beta}(|J|)} } \,  
   \frac{w(\n_2)}{{ Z_{\Lambda,\beta}(|J|)} } 
    (-1)^{\mathcal F^{\mathcal J}(\n_1) } 
\end{eqnarray}
Next, split the pairs of configurations according to whether $\m:=\n_1+\n_2$ is frustration free or not.  
Equivalently stated, the distinction depends on whether the multigraph representation of $\m$ supports a loop with odd flux through edges of $J_{x,y}<0$.   
The contribution from pairs of the first class vanishes since 
$ (-1)^{\mathcal F(\n_1) } $ changes sign under the  switch between $\n_1$ and $\n_2$ along such an odd loop.  Hence the above double sum equals the contribution to it of the terms for which $\n_1+\n_2$ is frustration free.  In that case also $\n_1$ is frustration free and thus, by the above lemma $(-1)^{\mathbb F(\n_1)} =1$.   Eq. \eqref{F_dup} readily follows. 

Adapted to  the correlation function, the above switching argument yields
\begin{eqnarray} \label{eq:2pnt}
 \langle \sigma_u \sigma_v \rangle_{\Lambda,\beta}^{\mathcal J} \!\!\!\!\! \!&& \!\!\!\!\!\!\!
  \langle \sigma_u \sigma_v \rangle_{\Lambda,\beta}^{\mathcal |J|} = 
   \notag  \\ 
 &=& \!\!\!\!\!
    \sum_{\substack{\n_1:\,\E\to \Z_+ \\ \partial \n_1= \{u,v\} }}
   \sum_{\substack{\n_2:\,\E\to \Z_+ \\ \partial\n_2=\{u,v\}}} \frac{w(\n_1)}{{ Z_{\Lambda,\beta}(J)} } \,  
   \frac{w(\n_2)}{{ Z_{\Lambda,\beta}(|J|)} }
   (-1)^{\mathcal F^{\mathcal J} (\n_1)} \id_{FF}[\n_1+\n_2]   \notag \\[1ex]  \notag
&=& \!\!\!\!\!
    \sum_{\substack{\m:\,\E\to \Z_+ \\ \partial \m= \emptyset  }}
   \sum_{\substack{\n \leq \m  \\  \partial\n=\emptyset}} 
   \frac{w(\n)}{{ Z_{\Lambda,\beta}(J)} } \,  
   \frac{w(\m-\n)}{{ Z_{\Lambda,\beta}(|J|)} }\,  
     sgn_J(u, v ; \m) \, \id_{FF}[\m]  \,  \notag \\    \\ \notag
 & = &  \frac 
 {\bbE^{\emptyset, \emptyset} \left( \id_{FF}[\m] \cdot  sgn_J(u, v ; \m) \right) } 
  {\bbE^{\emptyset, \emptyset} \big( \id_{FF}[\m]   \  \big) }  
  = \bbE^{\emptyset, \emptyset} \left ( sgn_J(u, v ; \m) \big \vert \mbox{$\m$ is FF} \right)   
\end{eqnarray} 
where $\m$ stands for $\n_1+\n_2$.  
\end{proof}

\subsection{Gibbs states with mixed boundary conditions}  \label{sec:bc} \mbox{} 

One is often interested in Gibbs states under the boundary conditions which amount to fixing the spins at  parts of the boundary, i.e.  as:
\be \label{mixed_boudary}
\sigma_u \stackrel{\bf \pm b.c.}{  :=} 
\begin{cases}  +1 & u\in \partial \Lambda_+\\  
-1 & u\in \partial \Lambda_-
\end{cases}
 \quad \mbox{or, for comparison,} \quad \sigma_u \stackrel{\bf +\, b.c. }{   := }    
\begin{cases}  +1 & u\in \partial \Lambda_+\\  
+1 & u\in \partial \Lambda_-
\end{cases}
\, .
\ee
To present how RCR can be adapted to such boundary conditions   
let us start with the more elementary case of $+$ boundary conditions, i.e. $\sigma_u=+1$ along all of $\partial \Lambda= \partial \Lambda_-\sqcup \Lambda_+$.  

In that case the main change  in the  RCR expansion \eqref{eq:8} is that the source constraint  on $\n$ does not apply along $\partial \Lambda$, since 
its spins are not averaged over.    Consequently, one gets the following modification of \eqref{eq:8}. 
\begin{multline}\label{eq:8++}
Z^{+}_{\Lambda,\beta}(J,h) := 
 \frac{1}{2^{|\Lambda|} }\sum_{\sigma: \Lambda \to \{-1,1\}}  \, \times \, \\ 
 \times \, \prod_{(x,y)\in \E_0} \exp(\beta J_{x,y}\sigma_x\sigma_y) 
\prod_{\substack{x \in \Lambda_0 \\  u\in \partial \Lambda} }\exp(\beta J_{x,u}\sigma_x\sigma_u) \prod_{u\in \partial \Lambda} \id[\sigma_u=+1]
 \\ 
=   \frac{1}{2^{|\partial \Lambda|} }\sum_{\substack{\n:\,\E\to \Z_+ \\ \partial \n \cap \Lambda_0 = \emptyset }} w(\n) \,, 
\hspace{5.7cm} 
\end{multline}
where $\Lambda_0=\Lambda \setminus \partial \Lambda$ is the set of unconstrained spins,  $\mathcal E_0$ the set of edges within that set, and  the usual  RCR constraint $\partial \n = \emptyset$ is replaced by:
\be\label{+dn}
\partial \n \cap \Lambda_0 = \emptyset \,.
\ee  

Turning now to the mixed boundary conditions, let us note that the Gibbs state in $\Lambda_0$ under the Hamiltonian \eqref{H_no_h} at the  $(\pm0$ boundary conditions  can equivalently be presented as under  $(+)$ boundary conditions, but  at modified couplings with the  signs of $J_{x,u}$  flipped at sites $u\in \Lambda_-$. 
  Handling that as were handled the negative couplings in   Section~\ref{sec:F_RCR} we get 
the following expression for the mixed boundary conditions: 

\be\label{eq:8+-}
Z^{\pm}_{\Lambda,\beta}(J,h) 
=   \frac{1}{2^{|\partial \Lambda|} }\sum_{\substack{\n:\,\E\to \Z_+ \\ \partial\n\subset \partial \Lambda }} w(\n) (-1)^{\mathcal F_{\partial \Lambda_-}(\n)} \,, 
\ee
where $\mathcal F_{\partial \Lambda_-}(\n) := \sum_{x\in \Lambda_0} \sum_{u\in \partial \Lambda_-} n( (x,u))$.  \\

Thus, mixed boundary conditions  induce the possibility of frustration.  However in this case the test is rather simple:  a configuration $\m$ carries frustration if and only if it supports a path linking $\partial \Lambda_+$ with $\partial \Lambda_-$     
Applying Theorem~\ref{RSR_FF} to the above setup, we get:

\begin{theorem}  \label{thm:FF}
For any finite Ising system with the ferromagnetic Hmiltonian \eqref{H_no_h} and the pair of boundary conditions spelled in \eqref{mixed_boudary}: \
\be \label{ZZ_RCR}
\frac{Z_{\Lambda,\beta}^{\pm}}{Z_{\Lambda,\beta}^{+}} = \bbP_{\Lambda,\beta}^{\emptyset,\emptyset;+} \left ( \partial \Lambda_-\stackrel{\n_1+\n_2}{ \, \not \! \! \!\longleftrightarrow}  \partial \Lambda_+
\right)  
\ee
and
\begin{multline} \label{eq:1pnt0}
 \langle \sigma_x   \rangle_{\Lambda,\beta}^{\pm} \,\, 
  \langle \sigma_x  \rangle_{\Lambda,\beta}^{+} =  
  \bbP^{\emptyset, \emptyset;+}_{\Lambda, \beta}\left ( u \stackrel{\n_1+\n_2}{\longleftrightarrow}
 \partial \Lambda_+ \big \vert \, FF \right)  - 
   \bbP^{\emptyset,\emptyset;+}_{\Lambda, \beta}\left ( u \stackrel{\n_1+\n_2}{\longleftrightarrow}
 \partial \Lambda_- \big \vert \, FF \right) 
  \end{multline}
  where the superscripts on $\langle -\rangle$ indicate the boundary conditions,  those on $\bbP$  indicates 
  $\partial \n_j \cap \Lambda_0 = \emptyset$ for $j=1,2$, and $FF$   denotes the event $\left(\partial \Lambda_-\stackrel{\n_1+\n_2}{ \, \not \! \! \!\longleftrightarrow}  \partial \Lambda_+ \right) $.
\end{theorem}
Analogous reasoning can be carried within the FK representation, for which one gets the following modified version of \eqref{eq:28}.
\be  
Z^+_{\Lambda,\beta}(J,h)  
=  \frac{e^{\beta \sum_{b\in \mathcal E} J_b}}{2^{|\Lambda|} }  
  \sum_{ \omega \subset \mathcal E}\,\, 
\prod_{b\in \omega}  p_b 
\prod_{b\in \mathcal E \setminus \omega} (1-p_b)  \cdot 
q^{N_0(\om)}
\ee  
where $N_0(\om)$ is the number of $\omega$-connected clusters which do not reach the boundary $\partial \Lambda$.

A similar expansion holds for the mixed boundary conditions, except that in this case  the consistency condition $\id_\omega[\sigma]$  fails if $\omega$ supports a path connecting $\partial \Lambda_-$ with $\partial \Lambda_+$.   Following this line of reasoning one  gets the following FK versions of \eqref{ZZ_RCR} and  \eqref{eq:1pnt0}
\be \label{ZZ_RCR_2}
\frac{Z_{\Lambda,\beta}^{\pm}}{Z_{\Lambda,\beta}^{+}} = \bbP_{\Lambda,\beta}^{\text{FK}(2),+} \left ( \partial \Lambda_-\stackrel{\omega}{ \, \not \! \! \!\longleftrightarrow}  \partial \Lambda_+
\right)  \,.
\ee
\begin{multline} \label{eq:1pnt_FK}
 \langle \sigma_x   \rangle_{\Lambda,\beta}^{\pm} \,\,  =  
  \bbP^{\text{FK}(2),+}_{\Lambda, \beta}\left ( u \stackrel{\omega}{\longleftrightarrow}
 \partial \Lambda_+ \big \vert \, FF \right)  - 
   \bbP^{\text{FK}(2),+}_{\Lambda, \beta}\left ( u \stackrel{\omega}{\longleftrightarrow}
 \partial \Lambda_- \big \vert \, FF \right) 
  \end{multline}
  where  $FF$   denotes the event $\left(\partial \Lambda_-\stackrel{\omega}{ \, \not \! \! \!\longleftrightarrow}  \partial \Lambda_+ \right) $.

  For completeness, and comparison,  let us add the following (simpler) formulae for the magnetization under the  (+) b.c.
\be  \label{eq:magnetization}
 \langle \sigma_x   \rangle_{\Lambda,\beta}^{+} \,\,  =  
  \bbP^{\emptyset,\emptyset;+}_{\Lambda, \beta}\left ( u \stackrel{\n_1+\n_2}{\longleftrightarrow}  \partial \Lambda  \right)  
  =  \bbP^{\text{FK}(2),+}_{\Lambda, \beta}\left ( u \stackrel{\omega}{\longleftrightarrow} \partial \Lambda  \right)  
\ee

\subsection{Potential applications} \mbox{}

The ferromagnetic and frustration free analog of \eqref{eq:2pnt} let to the observation that in that case Ising model's long range order is equivalent to percolation of the double current loop soup~\cite{Aiz82}, and that has given a direction towards various useful RCR-based relations.    Implications of the extension presented here to spin glass models have not yet been fully explored, or at least found  useful.  
However cancellation arguments similar to those presented above are of relevance in the analysis of:
\begin{itemize} 
\item expectation values of Kadanoff's  disorder operators~\cite{KadCev71} (Section \ref{sec:duality})
\item Wilson loop expectation values in the $\Z_2$ lattice gauge model (Section \ref{sec:LGM})
\item an RCR proof of the Ding-Song-Sun inequality~\cite{DSS23} (Appendix~\ref{sec:DSS}).
\end{itemize}  

\section{Graph duality and disorder variables}  \label{sec:duality}

In planar Ising models disorder operators are naturally associated with paths $\S$ on the dual graph.    
Corresponding to each is the operator $T_\mathcal S$ which transforms the  couplings 
$\{J_b\}_b\in \mathcal E$ by the following rule:\footnote{As was noted in~\cite{KadCev71}, the collection of disorder operators associated with  individual  edges offers one of the ways to realize the 2D Ising model's Kramers-Wannier duality~\cite{KraWan41}.}   
\be \label{T_S}
(T_\mathcal S J)_b = \begin{cases} 
- J_b & \mbox{if the edge $b$ crosses  $\mathcal S$} \\ 
 J_b & \mbox{otherwise}
\end{cases}
\ee 
In the extension of this notion to higher dimensions, graph duality continues to play a role.

Duality beyond the planar case  is  exemplified by  $\Z^d$, or more explicitly the corresponding cubical complex whose  $k$-cells,  at $k=0,1,2,..d$, are the graph's vertices, edges, plaquettes, etc... .  
In general, a  duality relation is expressed in a pairing between the $k$ cells of a given graph and the $d-k$ cells of its dual, which are consistent with the complex structure.

Such a duality relation is easily seen to exists between the naturally embedded $\Z^d$ and its diagonally shifted version  $(\Z+1/2)^d$, with  each paired duo intersecting  transversally at the mid-point.  (E.g., in the case of $d=3$ and $k=1$, the dual of a edge of $\Z^3$ is the plaquette of $(\Z+1/2)^d$ which it pierces.)  

Thus $\Z^d$ is self dual, in the sense of graph isomorphism.  Nevertheless,  in discussing the pair's simultaneous  realizations we  shall keep the distinction,  referring to one as  $\Z^d$ and to the other as   $(\Z^d)^*$. 

The dual pairing exists also between any finite subgraph of the form  $(-L,L)^d\cap \Z^d$ and a suitable subgraph of $(\Z^d)^*$. However in that case self duality is lost (at the boundary).  \\

\noindent {\bf Remark:}  We refrain here from presenting  the graph duality relation in more general homological terms.  These are essential for the  discussion of gauge models of other symmetry groups (cf.~\cite{DunSch25, ForFre25}),  and usually also accompanied by the orientation of edges).  Both are avoidable for  $\Z_2$ whose two group elements  satisfy $\sigma^{-1} \equiv \sigma$.    

In line with that, we  refer here to the boundary of a collection ($S$) of plaquettes, as the set of edges which belong to odd number of its elements, i.e. 
\be 
\partial S := \left\{ b\in \mathcal E \,  :\,   \prod_{p\in S} (-1)^{\id\left[ b\in \partial p \right] }    = -1 ]\right\}
\ee
One may note that this choice of terms  is  aligned also with the definition of $\partial \n$ in \eqref{d_n}.

We now turn more specifically to the Ising model on $\Z^3$ with the nearest neighbor interaction.   In that case, disorder operators 
are naturally associated with collections of plaquettes 
of  $(\Z^3)^*$, with $T_\mathcal S$ defined by the corresponding  higher dimensional version of \eqref{T_S} and its ``expectation value'' defined as 
\be \label{meanTS}
\langle T_{\S} \rangle_{\Lambda,\beta} := 
\frac 
{ Z_{\Lambda, \beta}( T_S( J) )}
{ Z_{\Lambda, \beta}(  J)} \, .
\ee

\begin{figure}[h]               \centering
       \includegraphics[width=0.25\textwidth=6]{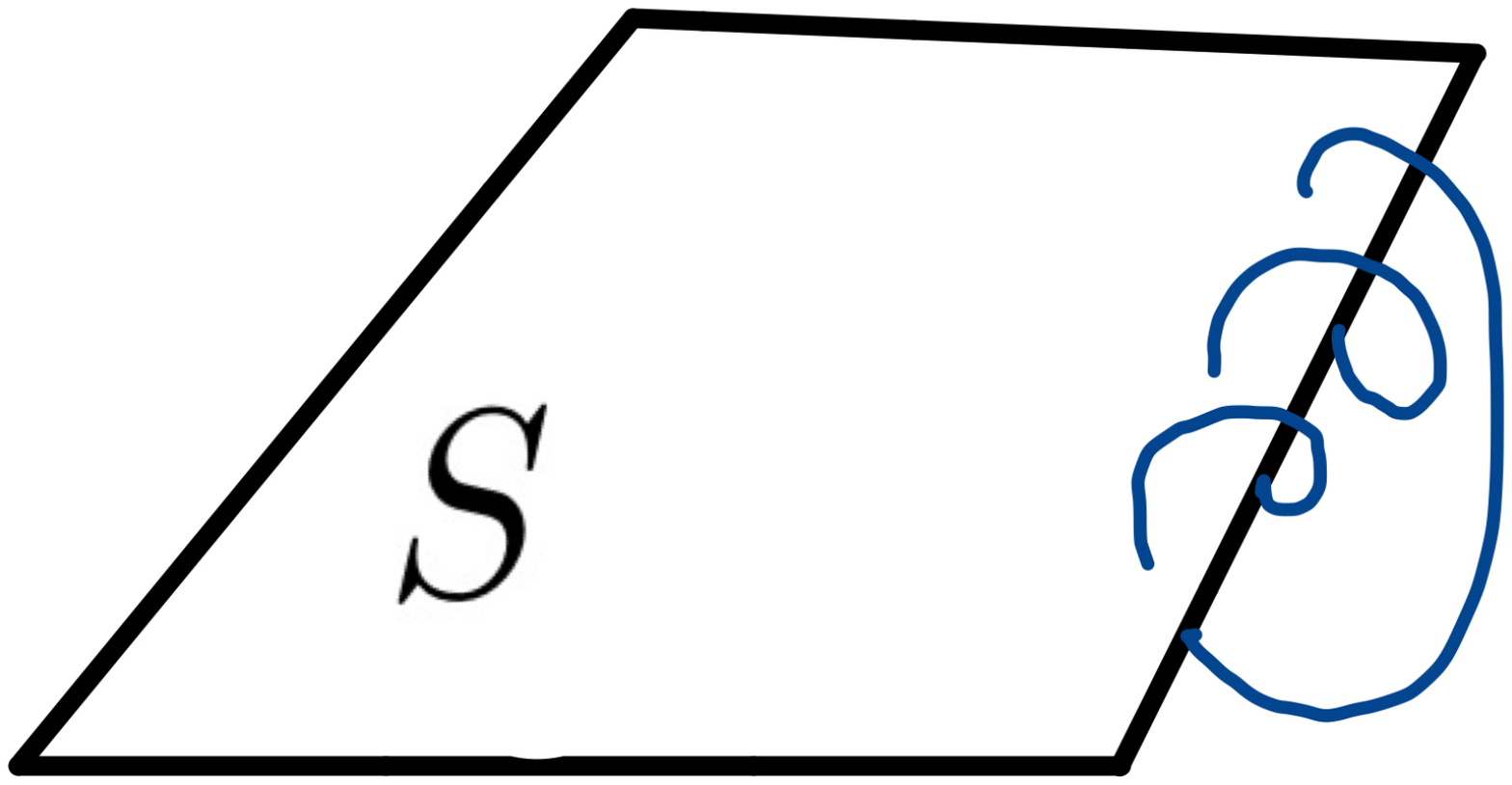} 
             \caption{ A  loop of odd winding through the surface $S$. 
             In \eqref{meanTS1} RCR configurations for  which  $\n_1+\n_2$  supports such a loop  make no contribution to  $\langle T_{S_\gamma}\rangle$.  A similar statement holds in  the model's FK representation.    \label{fig:XYZ}}
\end{figure}

From  the discussion in Section~\ref{sec:frustration} it readily follows that the disorder operator's expectation value \eqref{meanTS} admits the following probabilistic representation.

\begin{theorem} 
 For the Ising model on a finite domain $\Lambda \subset \Z^3$,  and any collection $S$ of dual  plaquettes of  $(\Z^3)^*$:
 \begin{eqnarray}  \label{meanTS1}
\langle T_{\S} \rangle_{\Lambda,\beta} 
&=& \bbP_{\Lambda,\beta}^{\emptyset, \emptyset} \left( 
\begin{array}{c}
\mbox{no  loop along  the support of  $(\n_1+\n_2)$} \\[1ex]   
\mbox{crosses $\S$ an odd number of times }
\end{array}  
\right)  \notag  \\  \\  \notag  
&=& \bbP_{\Lambda,\beta}^{\text{FK}(2)} \left( 
\begin{array}{c}
\mbox{no  loop along  the random edge set $\omega$ } \\[1ex]   
\mbox{ crosses $\S$ an odd number of times }
\end{array}  
\right) 
\end{eqnarray}  
\end{theorem} 
\noindent Where the support of  $\n_1+\n_2$ is the set of edges with $\n_1(b)+\n_2(b) \neq 0$, and the 
probabilities are presented  in the notation of Section~\ref{sec:frustration}: \\ 
 $\bullet$  $\bbP_{\Lambda,\beta}^{\emptyset, \emptyset}  $ refers to a duplicated set of  random currents with $\partial \n_1=\partial \n_2 = \emptyset$  \\ 
 $\bullet$  $\bbP_{\Lambda,\beta}^{\text{FK}(2)}$ refers to 
the Fortuin-Kasteleyn random cluster model, at $q=2$, \\ 
and the relevant condition  is depicted in  Fig. ~\ref{fig:XYZ}.

An example of interest is  the nearest neighbor Ising model  in 
\be \notag
\Lambda \equiv \Lambda_{L,K} = (-L,L]^2\times (-K,K]
\ee
 under the $+$ boundary conditions, and   $S$ the dual surface splitting $\Lambda_{L,K}$ at the level $x_1=1/2$.  In that case,  by a simple gauge transformation (flipping the spins over the lower part of $\Lambda_L$) one has
 \be \label{eq:ZZ}
 \langle T_{\S} \rangle_{\Lambda,\beta}^+  = \frac{ Z^{\pm }_{\Lambda.\beta}}  { Z^{+ }_{\Lambda.\beta}}  
  \ee  
  where $Z^{+}_{\Lambda.\beta}$ and $Z^{\pm}_{\Lambda.\beta}$ are the partition functions of the Ising model in $\Lambda_{L,K}$ with separately controlled boundary conditions on the upper and lower halves of $\Lambda_L$.

The above considerations shed some light on two related topics:  \\ 
1)   Equation~\eqref{meanTS1} provides a useful handle on the surface tension, and the related string tension in case the latter vanishes (Section~\ref{sec:ST}). \\ 
2) In view of the duality relation between the three dimensional Ising model and a $Z_2$  lattice gauge model,  \eqref{meanTS}  opens a path towards useful estimates on the Wilson loop variable in the latter (Section~\ref{sec:LGM}).  

Since Ising model's surface tension will  be invoked  in the discussion which follows, let us briefly recall here some  of the known results on this subject.

\subsection{Ising model's surface tension}   \label{sec:ST}  \mbox{}

The ferromagnetic Ising model on $\Z^d$ with the nearest neighbor interaction $J_{x,y}= J \delta_{\|x-y\|, 1}$ at low temperatures has two 
extremal translation invariant states.  These can be induced through the two extremal choices of the boundary conditions, one being  $\sigma = +1$ other $\sigma=-1$.    The surface tension between these two states  is defined
as~\footnote{The limit in \eqref{tau} exists irrespectively of the order in which $L,K \to \infty$, as can be seen through standard monotonicity arguments that are enabled by Griffith's inequalities.}:
\be \label{tau}
\tau(\beta) = \lim_{L\to \infty}  \lim_{K\to \infty} \frac{-1}{ (2L)^{d-1}}
\ln \left[
\frac{ Z^{\pm }_{\Lambda_{L,K}, \beta}}  { Z^{+ }_{\Lambda_{L,K}, \beta}}  \right]
\ee

 Existence of the limit defining $\tau$ was first proved in \cite{AGM73}.    For $\beta$ at which it is strictly positive:
 \be 
\frac{ Z^{\pm }_{\Lambda}(\beta)}  { Z^{+ }_{\Lambda}(\beta)}  \approx e^{- \tau(\beta)\, \text{area}(\Lambda) } 
\ee
(in a suitable sense) with $\text{area}(\Lambda) = (2L)^2$.  

A key result of  Lebowitz-Pfister~\cite{LebPfi81} (proved through Lebowitz inequalities~\cite{Leb74, Leb77} and other auxiliary statements cited there)   is that $\tau(\beta)>0$ holds over precisely the model's low temperature regime, as defined by symmetry breaking, or equivalently stated  the coexistence of two different translation invariant states.    
More explicitly : 
\begin{theorem}[~\cite{LebPfi81,BLP80}]  \label{thm:ST}
In the ferromagnetic n.n. Ising model on  $\Z^d$: 
 \be
   \tau(\beta) \leq  2 J  \beta M(\beta)^2 \quad \mbox{and}  \quad d\, \tau / d\, \beta  \geq  2 J M(\beta)^2  \, , 
\ee 
with $M(\beta) = \langle \sigma_0 \rangle^+_\beta$  the model's spontaneous magnetization and the derivative to be interpreted in the weak sense.  And by implication 
\be 
\tau(\beta) >0 \quad  \Longleftrightarrow \quad \beta>\beta_c \,\, \,. 
\ee 
\end{theorem} 

\section{The $\Z_2$  lattice gauge model} \label{sec:LGM}

\subsection{The pure gauge model and its phases} \mbox{} 

Lattice $\Z_2$ gauge models on $\Z^d$   have been of interest for a number of distinct reasons.    Balian-Drouffe-Itzykson~\cite{BDI75} pointed out that it provide a relatively simple setup for the study of  the mechanism which K. Wilson~\cite{Wil74} outlined towards an explanation of quark confinement.  It is also a special case of the more general class of binary spin systems with many-body interactions, which were discussed  earlier by F. Wegner~\cite{Weg71} in his  study  of duality  relations and phase transitions in generalized Ising models.   Some twenty years later the model's four-body interaction  terms reappeared as stabilizer operator in A. Kitaev's  quantum toric code~\cite{Kit03}, which provides an interesting example of  topological order and a mechanism for quantum error correction~\cite{CRSW25}.

Lattice gauge models (LGM) have recently drawn renewed attention, including the case of $\Z_2$ and other discrete symmetry groups~\cite{Cao20,Cha21,CaoCha23,DunSch25}. While this may not be the place for a thorough discussion of the physics and mathematics of gauge theories, we recall here that random current representation provides a useful tool  also for lattice gauge  models.

More explicitly, we complete here the proof that the $\Z_2$ lattice gauge model on $\Z^3$ undergoes a sharp transition, at which the expectation value of its Wilson loop variable transitions from area-law (fast) decay, to perimeter (slow) decay rate.   

The area law decay has already been established in the afore mentioned work of Lebowitz-Pfister~\cite{LebPfi81} to occur for (exactly) $\beta < \beta_c$, coinciding with the region over which the dual Ising model exhibits long range order.     
This is supplemented here with proof of a perimeter-law lower bound throughout the complementary phase,  in which the dual Ising model's spin-spin correlation's decay exponentially fast.    (The duality relation is recapitulated below.)  

It should be noted that sharpness of the analogous phase transition in the four dimensional self dual model was recently presented in \cite{DunSch25}.  In their analysis P. Duncan and B. Schweinhart extend the random surface analysis of \cite{ACCFR83} to $\Z_q$ gauge models in four dimensions,  paying close attention to the topological issues which arise there.  
These do not yet show in full in the three dimensional case discussed here.   
  
In the setup and the notation of~\cite{BDI75}, the fundamental variables of the $\Z_2$ gauge model on $\Z^d$ are associated with the graph's edges, and denoted here $A_{i,j} $ or $A_b$ (with $(i,j)$ or $b$ ranging over $\mathcal E$), with values   in $\Z_2$ (realized as $\{-1,1\}$).

In the pure gauge model the interaction terms are associated with the two-dimensional plaquettes of  $\Z^d$, to which we generically refer by the symbol $p$.   The model's Hamiltonian is:
\be \label{H_LGM}
H= - \sum_{p\in \mathcal P}  A_{i,j} A_{j,k}A_{k,l}A_{l,i }   \equiv - \sum_p \prod_{b\in \partial p} A_b  \,, 
\ee
where $\mathcal P$ is the collection of plaquettes of  $\Z^d$, and $\partial p$ denotes the boundary of $p$, which consists of four edges 
($b\in \mathcal E$).

The model's infinite volume version on $\Z^d$ is naturally 
constructed as the $L\to \infty$ limit of LGM formulated over $(-L,L]^d \cap \Z^d$.    
Convergence of the joint distribution of  $\{A_b\}_{b\in \mathcal {E}}$ can be deduced through monotonicity arguments based on the Griffith's second  inequality, which applies also under the multi-spin interaction \eqref{H_LGM}.   

The system's Hamiltonian and its  a-priori counting measure are invariant under the gauge transformations, which are mapping of the form  
\be 
A_{i,j} \mapsto     \tau_i A_{i,j}  \tau_j \,. 
\ee 
at an arbitrary site function $\tau: \mathcal V(\Z^d) \to \Z_2$.

Gauge invariant quantities of interest include the $\Z_2$ Wilson loop operators  
 \be  \label{A_gamma}
 A_\gamma  = \prod_{b\in \gamma} \, A_b\,,
 \ee
and their expectation values, the product being taken over the cyclic sequence of edges forming the loop ($\gamma$).

As a representative example, in the discussion which follows we shall focus on the $\ell\times \ell$  loops $\gamma_\ell$, each forming the  boundary  of the restriction of the  square $S_\ell = (1,L]^2\times^{d-2}\{0\}$ to the $\{x_3=0\}$ plane.   

In $d=2$ dimensions in the model's  Gibbs state equilibrium state the variables $\{A_p\}_{p\in \mathcal P}$ are  independently distributed 
and consequently the Wilson loop's expectation value is given by: 
\be 
\big\langle  \prod_{b\in \partial \gamma_\ell} A_b  \big\rangle_{\Z^2,\beta} = (\tanh \beta)^{\ell^2} 
\ee
with no sign of  phase transition in $\beta$.  

That however changes at higher dimensions.  In particular, for the model in three dimensions the following is true, and provable.  n(p)
 
\begin{theorem} \label{thm:LGM_sharp} The lattice gauge model on $\Z^3$, exhibits the following sharp transition from an area-law to a perimeter-law:
\begin{eqnarray}    
0 \leq \beta <\beta_c  &\Longrightarrow&   \big\langle  \prod_{b\in \partial \gamma_\ell} A_b  \big\rangle_{\Z^3,\beta}  \, \leq \, e^{- \alpha(\beta) \ell^2}  \quad  \mbox{ at $\alpha(\beta) >0$}
\label{conf}
   \\      
   \label{eq:deconf}
\beta > \beta_c  &\Longrightarrow&   \big\langle  \prod_{b\in \partial \gamma_\ell} A_b  \big\rangle_{\Z^3,\beta}\,  \geq \,  e^{-  m(\beta) \ell }    \qquad  \mbox{ at $m(\beta) < \infty$}
\end{eqnarray} 
The transition point $\beta_c $ is related by duality to the  Ising model's critical inverse temperature through  
\be \label{beta_duality} 
 \coth(\beta_c^{(\text{LGM})})  =  \exp \{2\beta^{(\text{Ising})}_c \}    \,.  
\ee
and the  asymptotic value of $\alpha(\beta)$ (for $\ell \to \infty$) coincides with the dual Ising model's surface tension (defined in \eqref{tau}). 
\end{theorem}

The  statements pertaining to the phase of area-low \eqref{conf} were proved by  Lebowitz-Pfister~\cite{LebPfi81} in the course of their study of the Ising model's surface tension (quoted above in Theorem~\ref{thm:ST}). 
Their analysis, aided by the other works which are cited there, relies crucially on Ising model's Lebowitz correlation inequalities~\cite{Leb77} and,  in transferring the results from the Ising model to the $\Z_2$ LGM, on the known duality relation~\cite{Weg71,BDI75} which is explained next.

Proof of the deconfinment lower bound \eqref{eq:deconf} is presented below, in Section~\ref{sec:deconf}.

 \subsection{Duality in three dimensions} \mbox{}

A key tool for the lattice gauge model's study is  \emph{duality}.  
In graph theoretic sense  (cf. Section~\ref{sec:duality}) plaquettes' dual objects in three dimensions are edges, and in four dimensions plaquettes of the dual lattice.  So it is natural to expect  the gauge model over $\Z^d$ to be dual to the Ising model in case of $d=3$ and self dual in the case of $d=4$.  

As is shown next the $\Z_2$ gauge model's random current representation provides a particularly convenient tool for the presentation of  this relation in three dimensions.    
 
\begin{theorem} 
For the $\Z_2$ lattice gauge model (LGM) on $\Z^3$: \\[1ex]   
\indent 1)  the  LGM's partition function in rectangular  subgraphs $\Lambda \subset \Z^3$ 
at inverse temperature $\beta$
 is simply related that of the Ising model on  dual $\Lambda^*$  at inverse temperature $\beta^*$ which is linked to $\beta$ by 
 \be   
e^{2\beta^*} = [\coth(\beta)]   \,.
\ee

\indent 2)  Under the above  correspondence, for each simple surface $\mathcal S$ of plaquettes of $\Z^3$ (within $\Lambda$) 
\be \label{Sdual}n(p)
\big\langle \prod_{b\in \partial \mathcal S} A_b \big\rangle^{(\text{LGT})}_{\Lambda,\beta} = \langle T_S \rangle^{(\text{Ising})}_{\Lambda^*, \beta^*}
\ee
where $\Lambda^*$ is the graph dual of $\Lambda$ and $T_S$ is as in \eqref{meanTS}, 
\end{theorem} 

\begin{proof}  
1)  Expanding the Gibbs factor as was done for the Ising model in \eqref{eq:8}  we are led  to an analogous representation of the gauge model's partition function, except that in this case the integer valued function $\n$ is defined over the set of plaquettes:
\begin{eqnarray}  \label{LGM_RCR}
Z^{(\text{LGM})}_{\Lambda}(\beta) &=& \frac{1}{ 2^{|\mathcal {E}|}} \sum_{A: \mathcal {E} \to \Z_2}  \prod_{p\in \mathcal {P}} e^{\beta A_{\partial p}} 
=  \frac{1}{ 2^{|\mathcal {E}|}}  \sum_{A: \mathcal {E} \to \Z_2}  \prod_{p\in \mathcal {P}} \left [ \sum_{n(p)=0}^\infty  
\frac{(\beta A_{\partial p})^{n(p)} }{n(p)!}  
 \right] 
\notag \\[2ex]  
 &=&  \sum_{\n:  \mathcal P  \to \Z_+} w(\n) \, \, \id\left[\partial \n = \emptyset\right]
\end{eqnarray}  
with weights
\be w(\n)  :=\prod_{p\in \mathcal P}\frac{ \beta^{\n(p)}}{\n(p) !}  \,.
\ee 
and 
\be \label{48}
\partial \n := \{b\in \mathcal {E}: (-1)^{\sum_{p: b\in \partial p} n(p)} =  \, -1\}\,.
\ee
That is: the boundary of a configuration $\n$ is defined as the set of edges $b$ for which the sum of the  fluxes of plaquettes whose boundary includes $b$ is odd.   
  
The path to duality is short: just notice that under the constraint \eqref{48},  for each $\n$ that contributes to the sum in \eqref{LGM_RCR}, and a $\pm 1$ value for $\sigma_{0^*}$ at a preselected dual site, there exists a unique spin configuration $\sigma: \mathcal{V} \to \{-1,1\}$ which is consistent with $\n$ is the sense that for any pair of neighboring dual sites $u,v\in \mathcal V^*$  and their separating plaquette $p$:
\be \label{eq:curl}
\sigma_u \sigma_v = (-1)^{n(p)}   \, .
\ee

To present the point in elementary term: the constraint \eqref{48} guarantees consistency of the condition \eqref{eq:curl} over any elementary loop of dual sites, which winds around an edge $b\in \mathcal E$.   Under the free boundary condition (though not under the periodic ones) any loop of the dual graph's can be presented as a symmetric difference (or $\Z_2$ product) of such elementary loops.  It follows that constructing a spin configuration through an iteration of \eqref{48} along an arbitrary path produces a globally consistent solution, which is easily seen to be unique modulo a global spin flip.   (This is not the case for periodic boundary conditions due to the presence there of non-contractable loops, i.e. the domain's non-trivial homotopy.)

Returning to the partition function, from where we left it at \eqref{LGM_RCR}: the partial sum of  $w(\n)$ over current configurations $\n$ which produce a \emph{specified spin configuration} $\{\sigma_x\}_{x\in \mathcal V^*}$ is the  product of 
$\cosh(\beta)$ for each nearest neighbor pair of equal spins ($n_{(u,v)^*}$  even) 
and  $\sinh(\beta)$ for pairs with opposite spins ($n((u,v)^*)$ odd).  Hence
\begin{multline}  \label{Z_coth}
Z^{(\text{LGM})}_{\Lambda}(\beta)  
= \frac 12 
 \sum_{\sigma: \mathcal {V}^* \to \{-1,+1\}} 
 \prod_{(u,v)\in {\mathcal E}^*} \left[\cosh(\beta)\right]^{(1 +\sigma_u \sigma_v)/2}
 \left[\sinh(\beta)\right]^{(1 -\sigma_u \sigma_v)/2}\\  
\\ 
= \frac 12  \big[ \cosh(\beta) \sinh(\beta)]^{|\mathcal {E}^*|/2}   
  \sum_{\sigma: \mathcal {V}^* \to \{-1,+1\}} \exp\left\{\sum_{(u,v)\in \mathcal {E}^*}  \beta^* \sigma_u \sigma_v \right\} 
\end{multline} 
with $\beta^*$ defined by the condition
\be
e^{\beta^*} = [\coth(\beta)]^{1/2}  \,.
\ee
Therefore, 
\be 
Z^{(\text{LGM})}_{\Lambda}(\beta)  = 
2^{| \mathcal {V}^* | -1}  \big[ \cosh(\beta) \sinh(\beta)]^{|\mathcal {E}^*|/2}   \,\, 
Z^{(\text{ISING})}_{\Lambda^*}(\beta^*)   
\ee 
from which the first statement follows.\\ 

2)  In expressing by means of \eqref{LGM_RCR}  the partition function modified through the insertion of  the Wilson loop variable $\prod_{p\in S} A_{\partial p}$,  the insertion's effect is equivalent to the flip of the  parity rule for  $n(p)$ of the plaquettes which lie in $S$, which translates into a flip  $\cosh(\beta) \Longleftrightarrow \sinh(\beta)$ over the edges which pierce $S$.  In terms of  \eqref{Z_coth} that is equivalent to the change of the sign in front of $\beta^*$ in the spin-spin interaction terms of dual edges which cross $S$.  Hence \eqref{Sdual}.  
\end{proof}

One may note here that the combination of  \eqref{Sdual}  and \eqref{eq:ZZ} links the gauge model's confinement parameter 
$\alpha(\beta)$ with the surface tension (defined in \eqref{tau}) of the dual Ising model.

\subsection{Proof of the deconfinement lower bound} \label{sec:deconf} \mbox{} 

At one step in the proof  it will be convenient to use the following elementary consequence of the convexity of the exponential function. 
\begin{lemma}  \label{lem:conv} For any $R \in (0,1)$ there exists a unique  $g>0$ at which 
$
1- R = e^{-g \cdot R}\,.
$ 
With $g(R)$ defined by that relation, for any $ r \in (0, R]$:  
\be 
1-r \geq e^{-g(R)\cdot r}
\ee 
\end{lemma}

\begin{proof}[Proof of Theorem~\ref{thm:LGM_sharp}:] \mbox{}   
  
Let $\gamma$ be a flat $\ell \times \ell$ loop on $\Z^3$as described below \eqref{A_gamma}, and  $S_\ell$  the square  spanned by it. 
By duality, for any  domain $\Lambda\subset \Z^3$  which is large enough to contain it,
\be
\big\langle  \prod_{b\in \partial \gamma_\ell} A_b  \big\rangle_{\Lambda,\beta}   = \big\langle T_{S_\gamma} \big \rangle_{\Lambda^*, \beta^*}
\ee  
where 
$\big\langle - \big \rangle_{\Lambda^*, \beta^*} $ is the expectation value over the Gibbs equilibrium state of the dual Ising model, formulated over $\Lambda^*$, at the inverse temperature $\beta^*$ (linked to $\beta$ by \eqref{beta_duality}).   

\begin{figure}[h]               \centering
       \includegraphics[width=0.3\textwidth=6]{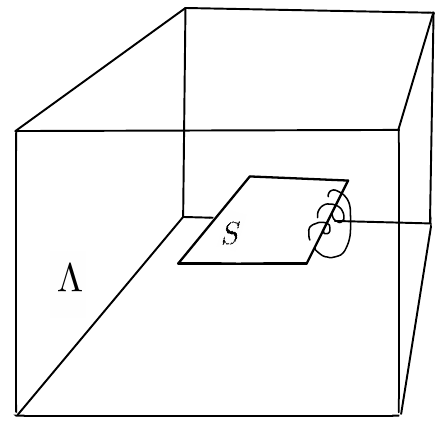} \qquad  \includegraphics[width=0.29\textwidth=6]{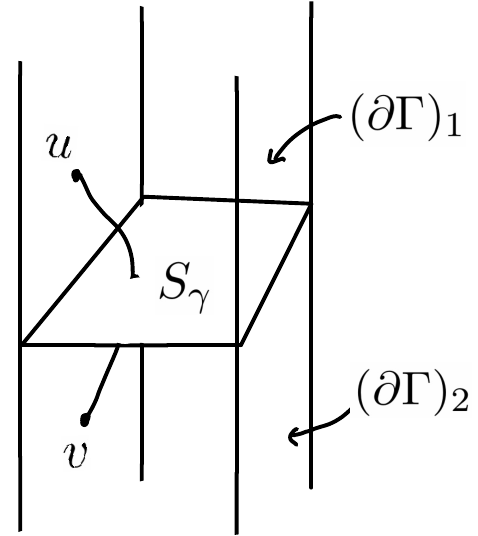}  
             \caption{The odd-winding event depicted on the left requires there being at least one  $u \leftrightarrow v$ connection between the upper and lower sides of the boundary of   the cylindrical set  $\Gamma$ constructed over $S$. Our lower bound on $\langle T_S\rangle $ starts from  an upper bound on the probabilities of such events.           \label{fig:uv}}
\end{figure}

Switching now to the Ising model's perspective, let us invoke \eqref{meanTS1} in its FK version.  
\begin{multline} \label{eq:string} 
\big\langle T_{S_\gamma} \big \rangle_{\Lambda^*, \beta^*} 
=  \\       
= \bbP_{\Lambda^*, \beta^*}^{\text{FK}(2)} \left( 
\begin{array}{c}
\mbox{each closed loop along  the random  } \\[1ex]   
\mbox{ edge set $\omega$ crosses $\S$ an even number of times }
\end{array}  
\right) 
\end{multline}  

For a convenient lower bound on the event stated above let us consider the cylindrical subset of $\Lambda^*$ 
which projects onto $S_\ell$, i.e. 
$\Gamma_\ell = \Lambda^* \cap \left[ S_\ell \times (\R+ 1/2) \right]$.   
The boundary of this set is cut by $S_\ell$ into two parts.  In line with the notation used above we denote 
the upper as $(\partial \Gamma_\ell)_1$ and the lower as $(\partial \Gamma_\ell)_2$.  

By the geometry of this setup (depicted in Fig.~\ref{fig:uv}), the parity condition in \eqref{eq:string} is certainly satisfied for any 
edge set $\omega$ for which the two parts of the cylindrical set's boundary are not $\omega$-connected.  Hence

\begin{eqnarray} \label{eq:string2} 
\big\langle T_{S_\gamma} \big \rangle_{\Lambda^*, \beta^*}  & \geq &  \bbP_{\Lambda^*,\beta^*}^{{\bf FK}(2)}  \left( 
(\partial \Gamma_\ell)_1  
\stackrel{\omega}{\,\, \, \not \! \! \!\longleftrightarrow}
(\partial \Gamma_\ell)_2 \right)  \notag \\  
 &=& 
\bbP_{\Lambda^*, \beta^*}^{{\bf FK}(2)} \left( 
\prod_{\substack {u\in   (\partial \Gamma_\ell)_1}}
\prod_{\substack {v\in   (\partial \Gamma_\ell)_2}}
\id \left[u    \stackrel{\omega}{\,\, \, \not \! \! \!\longleftrightarrow} v \right] \right) \,.  
 \end{eqnarray} 
The next (and key) step is to note  that by the FKG inequality (cf.~\cite{FKG71,Gri06}) under the FK random cluster measures (which have the FKG positive association property)
 the non-connection events, being non-increasing in $\omega$, are positively correlated.  

The probabilities of the individual non-intersection events are expressible in terms of the dual Ising model's two point function as:  
\begin{eqnarray} \label{eq:uv} 
\bbP_{\Lambda^*,\beta^*}^{{\bf FK}(2)} \left( 
u   \stackrel{\omega}{\,\, \, \not \! \! \!\longleftrightarrow} v \right)  &=&   1-\langle \sigma_{u} \sigma_{v} \rangle_{\Lambda^*, \beta^*}   
\notag \\ 
&\geq& 1- \langle \sigma_{u} \sigma_{v} \rangle_{(\Z^3)^*, \beta^*}
 \end{eqnarray} 

Combing the above observations  one learns that
\begin{eqnarray} \label{eq:string3} 
\big\langle T_{S_\gamma} \big \rangle_{\Lambda^*, \beta^*}  
 &\geq& 
\prod_{\substack {u\in   (\partial \Gamma_\ell)_1}}
\prod_{\substack {v\in   (\partial \Gamma_\ell)_2}}
\left[1-\langle \sigma_{u} \sigma_{v} \rangle_{(\Z^3)^*, \beta^*}  \right]
 \end{eqnarray}

At this point  it is convenient to invoke Lemma~\ref{lem:conv}.   
Denoting 
\be R(\beta^*) = \max_{u \neq v} \{  \langle \sigma_u \sigma_v \rangle _{(\Z^3)^*, \beta^*}\}  \quad \mbox{and \quad
$\lambda(\beta) = g(R(\beta^*))$}
\ee
we get
\begin{eqnarray} \label{eq:string4} 
\big\langle T_{S_\gamma} \big \rangle_{\Lambda^*, \beta^*}  
 &\geq&  
 \exp
 \left[- \lambda(\beta) \sum_{u\in   (\partial \Gamma_\ell)_1} \sum_{v\in   (\partial \Gamma_\ell)_2}
\langle \sigma_{u} \sigma_{v} \rangle_{(\Z^3)^*, \beta^*}  \right]
 \end{eqnarray}

The prefactor here is positive ($\lambda(\beta) >0$) for each $\beta >0$ since for each $\beta*<\infty$: $R(\beta^*)<1$, and hence $g(R(\beta^*) >0$.   So the relevant question is the magnitude of the double sum.   

By  the known sharpness of the Ising model's phase transition ~\cite{ABF87} we learn that for any $\beta^* > \beta^{Ising}_c$ the above two point function decays exponentially.  Under this condition a simple estimate shows that 
\be 
\sum_{u\in   (\partial \Gamma_\ell)_1} \sum_{v\in   (\partial \Gamma_\ell)_2}
\langle \sigma_{u} \sigma_{v} \rangle_{(\Z^3)^*, \beta^*}   \, \leq \,  \alpha(\beta) \, \text{per}(\gamma_\ell)
\ee
where the bulk of the sum is due to pairs of sites which are correlation length apart, on different sides of the curve $\gamma_\ell$, and 
$\text{per}(\gamma_\ell)=4\ell$ is the curve's perimeter.

This implies the deconfinement lower bound, that was claimed in \eqref{eq:deconf}.  References to the other statements made in Theorem~\ref{thm:LGM_sharp} were provided above. 
\end{proof}

Let us add here some remarks of technical nature. 
  
1)  While some of the key steps in the above proof was expressed in terms of the FK random cluster representation, they also could have been presented in terms of the RCR random current representation.   The main difference would be to replace the simple FKG argument in the justification of  \eqref{eq:string3} by a slightly more involved one which is available for RCR through subsequent conditioning and monotonicity arguments, of the type indicated here in Section~\ref{sec:path}.  

2)  The switch to RCR produces  a minor gain in that the probability of connection in RCR terms is the square of the corresponding one for FK (i.e. $\langle \sigma_u \sigma_v \rangle$ would be replaced by $\langle \sigma_u \sigma_v \rangle^2$).  This  is in line with the general thumb-rule  that was outline above:  when both RCR and FK techniques are applicable, which is not always the case,  the key  inequalities may be simpler to explain is the context of FK.  However tighter bound are provided by RCR.  In some cases it may be beneficial to combine the perspective which the representations  provide (see, eg. \cite{DGR20}).

3) At different points in the discussion we find that not only the percolation thresholds coincide but also probabilities of certain events, e.g. in \eqref{meanTS1},   are exactly the same for RCR and FK percolation.  Still may  other are not (e.g. the two point connectivity probability, as mentioned above).  
It could be of interest to clarify the relevant criterion.  More is said on the subject in the  recent study \cite{HJK25} where these and other graphical representations of a number of models are put on common footing.

\noindent And beyond  technicalities: 

4) The stochastic geometry underlying the  deconfinement transition have been laid open in the random plaquette version of this transition, which was presented and studied in \cite{ACCFR83,GriMar90}.  In it, the relevant stochastic-geometric  phenomenon is the spontaneous formation in a system of random plaquettes of surfaces which extend without boundary.  This can be viewed as a higher-dimensional analog of the 
the spontaneous formation, in the bond percolation model, of lines extending to infinity along open edges.    
In $3D$ the random surface transition is dual to percolation in the dual bond percolation model.  
Duality  relations extends further into more general $q$-state versions of the models, of which Ising and percolation  correspond to $q=2, 1$, and into higher dimensions where topology demands greater attention.   For interesting recent results on the deconfinement transition in this richer context  see~\cite{DunSch25}, and references therein.    

 \section{Extensions of RCR}   

The RCR  formulated above is for the spin-flip  symmetric Ising Hamiltonian \eqref{H_no_h} at zero magnetic field ($h=0$).  This formalism admits a number of extensions, some of which are by now well known, while some are less so:

\noindent \underline{i) Incorporation of external field.} 
 
 The external field terms 
$\{h_x \sigma_x\}$ are turned into binary spin-spin interaction  through the incorporation of a ``ghost spin'' $\sigma_{\Q}$  and replacement  $\sigma_x \rightarrow   \sigma_x \sigma_{\Q}$.  This stratagem was first employed  by R. Griffiths in his extension of  spin correlation inequalities to odd numbers of spins~\cite{Gri67}.  It's natural extension to random currents \cite{ABF87}  inspired the very useful addition of a ghost site in percolation models ~\cite{AizBar87}.

\noindent \underline{ii)  Other classes of variables.} 

Beyond Ising, to spin variables in the Griffiths-Simon class~\cite{Gri69,GriSim73}.  This enabled the extension of the RCR based results on the Ising spin model to the lattice $\phi^4_d$ field \cite{GriSim73, Aiz82}.  

\noindent \underline{iii)  RCR with local reflection symmetry}

The combinatorics underlying random current's switching symmetry require a pair of  partially matching RCR configurations.  In the presence of a reflection symmetry such a pair can be created by folding a domain onto itself, along a symmetry hyperplane.   Possible applications of this technique are  demonstrated here on two examples:

\begin{enumerate} 
\item The inequality of Messager and Miracle-Sole \cite{MesMir77} / Schrader~\cite{Sch77} on the monotonicity of the two point correlation function $S_2(0,\bf{x})$ in the nearest neighbor model on $\Z^d$.  

\item The van Beijeren inequality, which has the striking implication that the regular Ising model on $\Z^d$  exhibits translation symmetry breaking at all temperatures below $T_c(\Z^{d-1})$. 

 \end{enumerate}
Although the inequalities  are not new,  the conditional symmetry argument sheds additional light on the remainder.

\noindent \underline{iv) RCR for frustration-bearing interactions}

The adjustment of RCR to systems with frustration is discussed above in Section~\ref{sec:frustration}.   
As an example of application of the techniques presented there, in Appendix~\ref{sec:DSS} we show a simpler  proof of the recent (and previously unnoticed) inequality of Ding-Song-Sun\cite{DSS23}, and direct to yet  more recent applications of that inequality.  

Of the above, i) and ii) were presented and applied in \cite{Aiz82} and  iii) - iv) were within the original outline of  ``Part III''.

For completeness of the list let us mention in passing also the following more recent extensions of the random current method.

\noindent \underline{v) Infinite graphs.}  

A natural  extension of the RCR notions and basic properties to infinite graphs was developed and put to use in ~\cite{AizDumSid15}.  
That  enabled  direct applications of arguments invoking ergodicity (under space translations)  and uniqueness of the infinite cluster.  These played a role in the proof of the continuity of spontaneous magnetization at $T_c$ for  three dimensional version of the model.  \\[3ex] 

\noindent \underline{vi) Quantum Ising model in transversal field.}  

The quantum Ising model, over a graph $\mathcal G = (\mathcal V, \mathcal E)$ is a system of binary quantum spins, q-bits, attached to the sites of the vertex set $\mathcal V$, with the quantum Hamiltonian built of local Pauli spin triplets of operators $ (\sigma_u^x, \sigma_u^y, \sigma_u^z)$: 
\be
H = -\sum_{(u,v)\in \mathcal E} J_{u,v} \sigma^z_u \sigma^z_v - \eta \sum_{u\in \mathcal V} \sigma^x_u  
\ee 
(to which one may add the analog of the classical symmetry breaking term  term $(- h) \sum_{u\in \mathcal V} \sigma^z_u    $).  The above acts in the natural way on the tensor product Hilbert space $\mathbb X = \otimes_{u\in \mathcal V} \mathbb X_u$, with the local spaces isomorphic to $\C^2$.  
The extension of Random Current  to this model was developped in  works of Bj\"ornberg-Grimmett \cite{BjoGri09}, and Crawford-Ioffe \cite{CraIof10},  and used for the natural  extension to it of some of the fundamental results which were proven by RCR for the classical Ising model.      

\noindent \underline{vii) Random currents for XY models.} 

   In a new direction - a simple extension of the random current to two component classical spin systems $\underline {\sigma}_u = (\sigma^{(1)}_u, \sigma^{(2)}_u)$, with the $O(2)$ symmetric Hamiltonian 
\be
H = -\sum_{(u,v)\in \mathcal E} J_{u,v}\,\,  \underline {\sigma}_u \cdot \underline {\sigma}_v 
\ee 
was presented in \cite{EngLis23}.   It has not yet been shown to be as effective as it has in the discrete case.  But as there still are  some open challenges concerning the analysis of continuous spin models, we may still  hear more on this topic.   \\

\appendix 

\section*{Appendix: Unified derivation of some of the known correlation inequalities} 

In line with the originally outlined TOC, following is a RCR derivation of a number of Ising model's correlation Inequalities that, except for the first one, were originally derived by different means.  Where appropriate we start by explaining the general principle which is invoked in the proof.    

\section{Griffiths inequalities} 

\begin{theorem} [Griffiths I \& II] For any ferromagnetic Ising spin system which is invariant under spin flips, with a  Hamiltonian of the form 
\be 
H = - \sum_{\substack{B\subset \Lambda \\ |B| \mbox{even}} }J_B \prod_{u\in B} \sigma_u
\ee
at $J_B\geq 0$, 
\be
\label{eq:G1} 
\langle \sigma_A \rangle \geq 0   \\  
\ee
\be\label{eq:G2} 
\langle \sigma_{A_1} \sigma_{A_2} \rangle \geq \langle \sigma_{A_1}  \rangle  \, \langle \sigma_{A_2} \rangle 
\ee
\end{theorem} 
\begin{proof}   For $H$ with only pair interactions, on which we focused in Section~\ref{sec:RCR} Griffiths' first inequality hold by the positivity of \eqref{eq:erg}, which is how it was proved.  
The switching lemma proof of the second inequality,  which is not how it is often presented, 
is given in \eqref{pf:G2}.  As can be seen there, this short RCR  argument yields additional 
insights on the geometric nature of the spread of correlations. 

The proof carries verbatim also to the more general even Hamiltonians, in which case the RCR representation is expressed in terms of 
 $\n= \{n_B\}$, with $B$ ranging over   a larger collection of cells  (as is done here in the context of LGM,  Section~\ref{sec:LGM}).
 \end{proof}  

Also to be mentioned here   is the Griffiths-Hurst-Sherman GHS inequality~\cite{GHS70}, which  states that 
under ferromagnetic pair interactions:
\be \label{GHS}
\frac{-1}{2}\,  U_4(x_1,x_2,x_3,x_4) / 2 \, \geq \,    \langle 1, 2 \rangle\,   \langle 1, 3 \rangle 
\, \langle 1, 4 \rangle  \,. 
\ee  
It was actually in this context that the switching lemma made its appearance (at least as far as the author is aware).   However while for the mere positivity of $-U_4$ requires not more than that; a more refined analysis is required for the improved lower bound stated in \eqref{GHS}(cf.~\cite{Gra82})
  
A notable implication is the convexity in $h$ of the magnetization 
\be
\frac{\partial ^2 \langle \sigma_0 \rangle }{\partial h^2}  \leq 0, \quad \mbox{for all $h\geq0$}\, , 
\ee 
and that carries implications for the model's critical exponents~\cite{GHS70} (and provides and alternative proof of continuity of the liming function $M(\beta,h)$ in $h$, at $h\neq 0$. 

\section{Simon-Lieb inequalities} 
The following pair of inequalities conveys a geometric picture of the spread of spin correlations which has repeatedly been found to play a useful auxiliary  role in rigorous studies of critical phenomena.  The name refers to a pair of papers, the first by B. Simon~\cite{Sim80} and its follow-up with a useful improvement by E.H. Lieb~\cite{Lie80}.   In short time it was noted (cf. \cite{Lie80}) that the first inequality can be found within a larger collection of relations presented earlier by R. J. Boel and P.W. Kasteleyn~\cite{BoeKas78}.   

Analogous relations  hold for a variety of other such systems, including percolation~\cite{Ham57, AizNew84}, XY models~\cite{Riv80,AizSim80}, and  general stat mech systems of finite range interactions~\cite{DobPec83}, as well as random operators with finite range  hopping term, cf.~\cite{AizWar15}).   Of these, the RCR proof given below for the Ising case may be the simplest to present.  

\begin{theorem}[Simon~\cite{Sim80}, Lieb~\cite{Lie80}] \label{SL_ineq}
In the ferromagnetic Ising model over a finite graph $\G=(\mathcal V, \mathcal E)$, for any pairs of sites $x,y \in \mathcal V$ and a set of vertices such that any path from $x$ to $y$   has to include at least one site of $S$.  Then 
\be \label{eq_SL_V}
\langle \sigma_x \sigma_y \rangle \ \leq \ \sum_{u\in S} \langle \sigma_x \sigma_u \rangle_{\G_{S,x}}\, \langle \sigma_u \sigma_y \rangle
\ee
where 
$\G_{S,x}$ is the subgraph whose vertices include only the sites which are reachable from $x$ by paths which either do not visit $S$   or terminate on $S$ upon arrival, and whose edge set includes only the corresponding ``inner edges'', i.e. no edges between sites of $S$.

Likewise, given a pair $x,y\in \mathcal V$ and a set $B\subset \mathcal V$ which includes $x$ but not $y$
\be \label{eq_SL_E}
\langle \sigma_x \sigma_y \rangle \ \leq \ \sum_{\substack{u\in B \ v\in B^c} }
 \langle \sigma_x \sigma_u \rangle_B\, \tanh (\beta  J_{u,v} ) \,   \langle \sigma_v \sigma_y \rangle 
\ee
with $ \langle -- \rangle_B$ the Gibbs state for which the interactions are limited to edges in $B$.

\end{theorem} 

\begin{proof}  To prove \eqref{eq_SL_V} we start with the random walk representation of Theorem~\ref{thm:path_exp}) 
by which $\langle \sigma_x \sigma_y\rangle$ is expressed in term of the random path starting at $x$ and terminating upon reaching $y$.   

Decomposing  each path according to the site  $u$ at which the path hits the vertex set $S$ we get:
\begin{eqnarray}  \label{eq_SL_Vpath}
\langle \sigma_x \sigma_y \rangle &=& \sum_{\gamma: \, x \to y} \rho(\gamma) \ = \ 
\sum_{u\in S} \, \sum_{\substack{\gamma_{x,u}: \, x \to u \\ 
\gamma_{u,y}: \, u \to y}}   \rho(\gamma_{x,u} \circ  \gamma_{u,y})  \notag \\[2ex]  
&=& \sum_{u\in S} \, \sum_{\substack{\gamma_{x,u}\\ \mbox{\tiny hitting $S$ at $u$}}} \rho(\gamma_{x,u})  \langle \sigma_x \sigma_u\rangle_{\widehat \gamma_{x,u}^c}  
 \\[2ex]  & \leq&   \sum_{u\in S}  \sum_{\substack{\gamma_{x,u}\\ \mbox{\tiny hitting $S$ at $u$}}} \!\rho_{_{\G_{S,x}}}(\gamma_{x,u}) \ \   \langle \sigma_x \sigma_u\rangle  
\ =\   \sum_{u\in S} \langle \sigma_x \sigma_u \rangle_{\G_{S,x}}\, \langle \sigma_u \sigma_y \rangle \notag 
\end{eqnarray}  
where: the transition to the second line is by the combination of the random walk properties (i) and (ii), and the inequality is by a combination of property (iii) (applied to $\rho(\gamma_{x,u})$ and the Griffiths monotonicity  applied to the two point correlation function. 

The proof of \eqref{eq_SL_E} is by a similar argument, in which the path is decomposed into 
$\gamma = \gamma_{x,u} \circ (u,v) \circ \gamma_{v,y}$ with $(u,v)$ the edge over which $\gamma$ exists $B$ for the first, but not necessarily the last time.   
\end{proof}

It may also be instructive to consider the following alternative RCR proof of \eqref{eq_SL_V} (the site version), which demonstrates a more rudimentary  use of switching lemma technique.

\begin{proof}[An alternative RCR proof of \eqref{eq_SL_V}]   \mbox{}\\[-2ex] 

We start from the RCR expression of the 
spin-spin correlation in terms of the  random currents $\n_1$ and $\n_2$ over the pair of  graphs $\G$ and $G_{S,x}$ correspondingly, where the sum over $\n_2$ being initially redundant, as it adds up to $1$)
  \begin{eqnarray} \label{eq:SL_E1}
  \displaystyle
\langle \sigma_x \sigma_y \rangle & = & 
\sum_{\substack{
\partial\n_1[\G]=\{x,y\} 
\\  
\partial\n_2[\G_{S,x}]=\emptyset}}  
\frac{w(\n_1)}{Z(\G)}\frac{w(\n_2)}{Z(\G_{S,x})} \, . 
\end{eqnarray} 
Under the stated source constraint  there has to be at least one simple path linking $x$ with $y$ along edges $(b)$ with $n_1(b)+n_2(b)>0$.  Each such paths hits $S$ at some  $u\in S$.  Summing the corresponding probabilities 
one gets:
\be
  \langle \sigma_x \sigma_y \rangle
\leq 
\!\!\sum_{u\in S} \sum_{\substack{
\partial\n_1=\{x,y\} 
\\  
\partial\n_2=\emptyset}}  \!\!\frac{w(\n_1)}{Z(\G)}\frac{w(\n_2)}{Z(\G_{S,x})} \,  
\ind\left[x \stackrel{\n_1+\n_2}{\longleftrightarrow}  y  \mbox{ in $\G_{S,x}$} \right] 
\ee
Applying to the above expression the switching lemma's relation \eqref{eq:switching} one gets, without further degradation of the bound, 
\be  
 \langle \sigma_x \sigma_y \rangle
\leq 
\sum_{u\in S} \sum_{\substack{
\partial\n_1=\{u,y\} 
\\  
\partial\n_2=\{x,u\}}}  \frac{w(\n_1)}{Z_1}\frac{w(\n_2)}{Z_2} \,  
\ = \  \sum_{u\in S} \langle \sigma_x \sigma_u \rangle_{\G_S} \, \langle \sigma_u \sigma_y  \rangle
\ee
which proves \eqref{eq_SL_V}. \end{proof}   

We omit here the  RCR proof  of the edge version of the above relation  \eqref{eq_SL_E} as in these terms it is less straightforward, though  doable.

\section{Ding-Song-Sun inequality} \label{sec:DSS}

J. Ding, J. Song, and R. Sun presented in ~\cite{DSS23} a new inequality, which has since found some interesting applications (listed below).    
The setup is that of the Ising spin model on a finite graph $\G = (\mathcal V, \mathcal E)$ under a magnetic field which is a sum of two term, one  non-negative and the other of general sign, i.e.
\be 
- H(\sigma) = \sum_{(u,v)\in \mathcal E} J_{u,v} \sigma_u \sigma_v + \sum_u (h_u + g_u) \sigma_u  \, ,
\ee  
with 
\begin{eqnarray}
g: \mathcal V \to \bbR\, ,  \qquad h: \mathcal V \to \bbR_+   
\quad \mbox{(i.e.  $h_u \geq 0 \quad \forall u\in \mathcal V $)}   \,. 
\end{eqnarray}

 The DSS inequality is:
\begin{theorem}[\cite{DSS23}] \label{thm:DSS}
In the above setup,  
\be \label{eq:DSS1}
\langle \sigma_x\rangle_{g+h}  - \langle \sigma_x\rangle_{g-h}  \leq 
\langle \sigma_x\rangle_{h}  - \langle \sigma_x\rangle_{-h} \, .
\ee  
And, by implication, 
\be  \label{eq:DSS2}
\langle \sigma_x ; \sigma_y \rangle_g  \leq \langle \sigma_x ; \sigma_y \rangle_g  \, ,
\ee 
where on the left  is the truncated correlation function
\be
\langle \sigma_x ; \sigma_y \rangle :=\langle \sigma_x  \sigma_y \rangle - \langle \sigma_x  \rangle \, \langle \sigma_y  \rangle \, . 
\ee \end{theorem} 

The pivotal statement here is the first one.   The second can be deduced from it simply through the expansion to second order in $h$, around   $h=0$.  However, the proof of \eqref{eq:DSS1}  is considerably more involved.     Following is what may potentially be seen as a simpler proof (admittedly a relative term), based on the random current techniques.

 \begin{proof}[Proof of Theorem~\ref{thm:DSS}]
Flipping the spins in the second expression on both RHS and LHS, the claimed \eqref{eq:DSS1} can be rewritten as
 \be \label{eq:DSS3}
\langle \sigma_x\rangle_{h+g}  + \langle \sigma_x\rangle_{h-g}  \leq 
2 \langle \sigma_x\rangle_{h}   
\ee  
   
For a convenient representation of \eqref{eq:DSS3}, let us construct a multigraph  
$\widehat {\mathcal G}$ through the following steps:  \\ 
\begin{itemize}
\item[1)]  Split the vertex set $\mathcal V$ according the sign of $g(x)$ into    
\be 
\mathcal {V}_+  = \left \{ u\in \mathcal V \, :\,  g(u) \geq 0 \right\} \, ,   \quad \mbox{and}  \quad 
\mathcal {V}_-  = \left \{ u\in \mathcal V \, :\,  g(u) < 0 \right\} \,.
\ee 
\item[2)]  Add  a  ``ghost site'' $\Q$ (\`a la Griffiths), thus forming   $\widehat {\mathcal{V}} = \mathcal V \cup \{\Q\}$.

\item[3)] Link $\Q$  to each of the original sites by  two  edges, one with coupling $J = h(x)$, and the other with $J =g(x)$.  \\  
 The collection of the added edges  of the first kind will be referred to as $\mathcal E_h$, and that of the second kind split into $\mathcal E_+\sqcup  \mathcal E_-$ according to the sign of $g(x)$.  
 \end{itemize} 
\noindent The construction results in an enhanced graph  $\widehat {\mathcal G} = ( \widehat{\mathcal V},  \widehat{\mathcal E})$ with the edge set:
 \be
 \widehat{\mathcal E} =   \mathcal {E}\sqcup  \mathcal E_h \sqcup  \mathcal E_+ \sqcup  \mathcal E_- \, . 
 \ee

The claimed relation \eqref{eq:DSS3}, reformulated over  $\widehat {\mathcal G}$ reads as
\be \label{eq:DSS4}
\langle \sigma_x \sigma_{\Q} \rangle_{h+g}  + \langle \sigma_x \sigma_{\Q}\rangle_{h-g}  \leq 
2 \langle \sigma_x \sigma_{\Q}\rangle_{h}   
\ee 
where the subscripts refer  to the pair of Hamiltonians  
\be 
- H(\sigma)_{\stackrel{+}{(-)}} = \sum_{(u,v)\in \mathcal E} J_{u,v} \sigma_u \sigma_v + \sum_u (h_u  \mbox{}_{\substack{+\\ (-)}}  g_u) \, \sigma_u \sigma_{\Q}   \,.  
\ee  
 
Let now $Z_+$ and $Z_-$ denote the partition functions over $\widehat {\mathcal G}$ under 
$H_+$ and $H_-$ correspondingly, and $\widetilde {Z}$ the one under the Hamiltonian with $g(x)$ replaced by its absolute value $|g(x)|$.  

In the ``frustration-adapted'' RCR representation of Section~\ref{sec:frustration} we have the following variant of \eqref{eq:2pnt} 
\begin{eqnarray}
\langle \sigma_x \sigma_y \rangle_{h+g} &=&  
 \sum_{\substack{ \partial \n_1= \{x,\mathfrak{g}\} \\ \partial \n_2= \emptyset  }}
  \frac{w(\n_1)}{{ \widetilde{Z}} } \,  
   \frac{w(\n_2)}{{ \widetilde{Z}} }
   (-1)^{\mathcal F_{\mathcal E_-} (\n_1)} \id_{FF}[\n_1+\n_2]   \notag \\  \\   \notag
   &=&  
 \sum_{\substack{ \partial \n_1= \{x,\Q\} \\ \partial \n_2= \emptyset  }}
  \frac{w(\n_1)}{{ \widetilde{Z}} } \,  
   \frac{w(\n_2)}{{ \widetilde{Z}} }
   sgn_J(u, v ; \n_1) 
   \id_{FF}[\n_1+\n_2] 
\end{eqnarray}
where: \\ 
\indent i)  ${\mathcal F_{\mathcal E_-} (\n_1)}$ is the total flux of $\n_1$ over the edges $\mathcal E_-$, i.e. the  sum over $\n_1$ on the edges linking   $\Q$ with the vertex set $\mathcal V_-$.   \\ 
\indent 
\indent  i) The frustration free condition $\id_{FF}[\n_1+\n_2] $ is that $\n_1+\n_2$ has no  loops of odd flux.\\  
\indent iii) $ sgn_J(u, v ; \n_1) $ is the parity of the connection defined in \eqref{sgn_omega} (under the Frustration Free condition FF).    

The corresponding representation for $\langle \sigma_x \sigma_y \rangle_{h-g} $ (i.e. at flipped $g$) has a similar representation, with   flux over  $\mathcal {E}_-$  replaced by that over  $\mathcal {E}_+$.  

In both cases the source conditions imply $x \stackrel{\n_1+\n_2}{\longleftrightarrow}  \Q$.  At given $\m=\n_1+\n_2$ simple  paths, along $m\neq 0$,  from $x$ to $\Q$ can arrive with an edge of 
$\widehat {\mathcal E}$,  $\mathcal {E}_+$ or $\mathcal {E}_-$.  
However, crucially,  in each case under the relevant FF condition for specified $\m$ the last two options cannot be realized simultaneously. 

 Thus, due to the factor $sgn_J(u, v ; \n_1)$,  configurations $(\n_1, \n_2)$ in which there is a path from $x$ reaching $\Q$ through an edge in $\mathcal {E}_+ \cup \mathcal {E}_-$ make exactly opposite contributions to 
$\langle \sigma_x\rangle_{h\pm g}$.     

Taking the above into account one learns that: 
\begin{multline}  \langle \sigma_x \sigma_{\Q} \rangle_{h+g}  + \langle \sigma_x \sigma_{\Q} \rangle_{h-g} = 
2 \sum_{\substack{ \partial \n_1= \{x,\Q\} \\ \partial \n_2= \emptyset  }}
  \frac{w(\n_1)}{{ \widetilde{Z}} } \,  
   \frac{w(\n_2)}{{ \widetilde{Z}} }   \\ 
 \times \id[\mbox{all paths $x \stackrel{\n_1+\n_2}{\longleftrightarrow}  \Q$ terminate with an edge in  $\mathcal {E}_h$}] \times \\ 
 \times \id[\mbox{there is no  path $\Q \stackrel{\n_1+\n_2}{\longleftrightarrow}  \Q$  starting by  $\mathcal {E}_-$ and ending by $\mathcal {E}_+$}] 
\end{multline} 

\be
\langle \sigma_x \sigma_{\Q} \rangle_{h+g}  + \langle \sigma_x \sigma_{\Q} \rangle_{h-g} \leq  
2 \sum_{\substack{ \partial \n_1= \{x,\Q\} \\ \partial \n_2= \emptyset  }}
  \frac{w(\n_1)}{{ \widetilde{Z}} } \,  
   \frac{w(\n_2)}{{ \widetilde{Z}} }   = 2 \langle \sigma_x \sigma_{\Q} \rangle_{h}
\ee
as was claimed in \eqref{eq:DSS4}.
 \end{proof}

The DSS inequality was recently of help in two unrelated works:  \\ 
Bauerschmidt-Degallier~\cite{BauDeg24} on the Log-Sobolev inequality for near critical Ising models, and 
Aizenman-Warzel~\cite{AizWar25} on the magnitude of entanglement  in the ground states of the quantum Ising models with transverse field. 

\section{RCR and reflection symmetry} \label{sec:reflection}

 Reflection symmetry, and reflection positivity, enable many shortcuts in problems of probability, analysis, statistical mechanics and field theory, and for some even play a fundamental role (\cite{OS73,GRS75,FILS78}).   

In the two examples which are  discussed next  we encounter situation to which  the following  definition applies.

\begin{definition}[Markovian reflections]  Let $\mathcal R_\X$ (abbreviated below by $\mathcal R$) be a reflection of $\R^d$ with 
with respect a hyperplane $\mathbb X$ of co-dimension $1$, and $\G$ a graph embedded in $\R^d$.   The graph   is said to have $\mathcal R$ as  a Markovian reflection symmetry if
\begin{enumerate} 
\item   the graph is invariant under the action of $\mathcal R$ 
\item  $\X$ does not cut through any of the graph's edges (i.e. each edge of $\G$ may either not intersect with $\X$, or be fully contained in it).
\end{enumerate}
An Ising spin system with pair interaction, as in in  \eqref{H_w_h}, is said to have $\mathcal R$ as  a Markovian  reflection symmetry if the graph on  which it is formulated has this property and its couplings are reflection invariant.
\end{definition} 

For Ising systems for which  condition (2) fails, but only  due to the presence of edges which are cut by $\X$  at mid-points and with orthogonal intersections, the Markovian condition can be met upon the addition of mid-edge spin variables with couplings which mediate the interaction between the two at the edge's boundary,  as indicated in Fig.~\ref{fig:folded} (on the right).  
 
When present, Markovian reflection symmetry combines well with the random current representation.  Through the pairing of the current configuration $\n$ with it reflection $\mathcal R_\X (\n)$ emerges a conditional symmetry, which enables  many of the familiar benefits of RCR duplication.  In particular one can reap the benefits of the switching lemma also in this context.
An example of that is the following statement with is proved below.  

\begin{theorem}  \label{thm_folded}  Let $\G=(\Lambda,\E)$ and $H$ be a finite graph embedded in $\R^d$ and an Ising spin pair interaction of the form \eqref{H_no_h} with a Markovian symmetry $\mathcal R$.   
Than for any pair of sites  $x, y \in \Lambda$ lying on the same side of the reflection hyperplane  $\mathbb X$
\be\label{eq_fold_prob}
\langle \sigma_x \sigma_{\mathcal {R}(y)} \rangle_{\Lambda, \beta} = \langle \sigma_x \sigma_y \rangle_{\Lambda, \beta} \, \, 
\mathbb P_{\Lambda, \beta}^{\{x,y\}}
\!\!\left( 
y \stackrel{\n+\mathcal R_{\mathbb X}(\n)}{\longleftrightarrow}  \mathbb X  \cap \Lambda
\right) 
\ee
where $\mathbb P_{\Lambda, \beta}^{\{x,y\}}$ is  the normalized joint probability distribution of the random current $\n$ of weights $w(n)$  under the source  constraint $\partial \n=\{ x,y\}$, as  in \eqref{w(n)}.
\end{theorem}

\begin{figure}[h]              
\centering
       \includegraphics[width=0.25\textwidth]{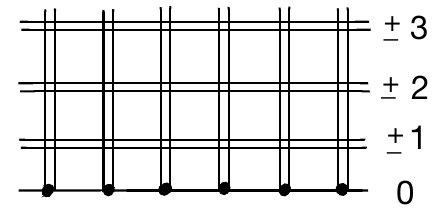} \qquad 
      \includegraphics[width=0.22\textwidth]{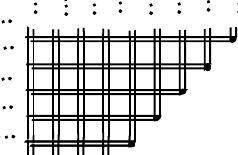} \qquad
            \includegraphics[width=0.25\textwidth]{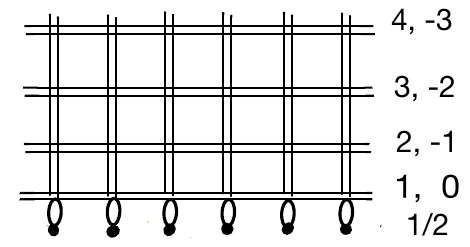} \qquad 
         \caption{Foldings of $\Z^2$ through reflection symmetries with respect to lines defined by: i) $x_2=0$, ii) $x_1+x_2=0$, and iii) $x_2=1/2$.  When the reflection plane cuts  through edges, we add at the cut sites intermediary local spin variables, as indicated on the right.   \label{fig:folded}  }
\end{figure}

In the proof we shall use the following notation and a reflection adapted version of the switching lemma.

For a given choice of the symmetry hyperplane, which (upon enhancement, if needed) cuts only through sites,  the graph's vertex set is presentable as a disjoint union:
\be \label{split}
\Lambda = \Lambda_0  \sqcup \Lambda_1 \sqcup \Lambda_2   \,,
\ee
with $\Lambda_0= \Lambda\cap \mathbb X$, and the other two sets 
containing the remaining sites split according to the side of $\mathbb X$ in which they lie.
Correspondingly we also split the edge set into: 
\be \label{D_edge}
\mathcal E = \mathcal E_0  \sqcup \mathcal E_1  \sqcup \mathcal E_2
\ee  
with $
 \mathcal E_0 := \{(u,v) : u,v \in \Lambda_0 \} 
$
and  for $j=1,2$: 
\be
 \mathcal E_j := \{(u,v) : u,v \in \Lambda_j \}  \cup \{(u,v) : u \in \Lambda_0, v\in \Lambda_j \}\,.
\ee  

The Markovian condition ensures that $\Lambda_0$ is a ``cut set'' for the connection $\Lambda_1 \longleftrightarrow \Lambda_2$, in the sense that any path along edges of $\mathcal E$, which starts at some $x\in \Lambda_1$ and ends at some  $y\in \Lambda_2\setminus \Lambda_0$ needs to pass through it.   


\begin{lemma}\label{lem:switch_reflect}   
Let $\G=(\Lambda,\E)$ be a graph with a Markovian reflection symmetry $\mathcal R_\X$, and let $A\subset \Lambda$, and $y \in \Lambda_1$. 
 Then  for any  Ising spin system on $\G$ with a symmetric  pair interaction of the form  \eqref{H_no_h}, and a flux configuration $\m: \E\to \Z_+$ such that  
$y \stackrel{\m}{ \longleftrightarrow} \Lambda_0$
 \begin{multline} \label{eq:switch_reflect}
\sum_{\substack{\n:\E\to \Z_+ \\ \partial\n=A}}
  w(\n)    \cdot  \id\big[\n+\mathcal R(\n) = \m\big]  \,  
 =  
 \!\! \!\!\sum_{\substack{\n:\E\to \Z_+ \\ \partial\n=A\Delta \{y, \mathcal {R}(y) \} }}
  \!\! \!\! \!\!
  w(\n)   \cdot  \id\big[\n+\mathcal R(\n) = \m\big]  \, . \qquad \mbox{} 
\end{multline}  
 \end{lemma}  
What should be noticed here is the difference in the two source condition in \eqref{eq:switch_reflect} and, importantly, the restriction of the statement to  $\m$ such that $y \stackrel{\m}{ \longleftrightarrow} \Lambda_0$.

\begin{proof} The stated relation holds trivially in case $y\in \Lambda_0$, so we shall  focus on the case $y\in \Lambda_1\setminus \Lambda_0$.   
Under the  assumption made on $\m$, there exists a simple path supported on edges with $\m(b) \neq 0$ which starts at $y$ and stops upon reaching $\Lambda_0$.  Concatenating it with its reflection, i.e. its image under  $\mathcal R$, we get a simple path $\gamma$ from $y$ to $\mathcal R_\X(y)$ such that: 
\begin{enumerate} 
\item[i)]   the path $\gamma$ visits $\Lambda_0$ exactly ones, arriving by an edge of $\E_1$ and departing by an edge of $\E_2$, 
\item[ii)]  it advances only over edges for which $\m + \mathcal R(\n) \neq 0$.
\end{enumerate}

Such a path, selected from the collections of paths with these properties as a measurable function of $\m$ and $y$,  can be used for a version of the familiar switching argument, this time within the pair  $\n$ and $\mathcal R(\n)$.   

To make that more explicit, each flux functions $\n, \m: \E \to \Z_+$ 
can be presented as a triple
\be \label{3n}
\n=(\n_0, \n_1, \n_2) 
\ee  with $\n_j$ the restriction of $\n$ to $\E_j$, and doing the same for $\m$.  In these terms, the claimed relation \eqref{eq:switch_reflect} takes the form:
\begin{multline} 
 \sum_{\n_0:\E_0\to \Z_+} w(\n_0) \delta_{\n_0, \m_0}  
     \!\!\! \sum_{\substack{\n_1:\E_1\to \Z_+ \\ \n_2:\E_1\to \Z_+ \\ \partial\n=A }}
w(\n_1)  w(\n_2)       \cdot  \id\big[\n_1+\mathcal R(\n_2) = \m_1\big]   \\  
 =  \sum_{\n_0:\E_0\to \Z_+} w(\n_0) \delta_{\n_0, \m_0}  
     \!\!\! \sum_{\substack{\n_1:\E_1\to \Z_+ \\ \n_2:\E_1\to \Z_+ \\ \partial\n=A\Delta \{y, \mathcal R(y) }}
 w(\n_1)  w(\n_2)       \cdot  \id\big[\n_1+\mathcal R(\n_2) = \m_1\big]  
\end{multline}
(where the only difference is in the source conditioin).
 The claimed \eqref{eq:switch_reflect}  follows by the combinatorial argument used above in the proof of Lemma~\ref{lem:switching}, applied here to the double sum at given $\n_0$.  

More explicitly:  one may start by restating the claimed relation in terms of  the multigraph $\G_1(\m)$, over the edge set $\E_1$ of edge multiplicity $\m_1$.
Simple combinatorics allow to read from \eqref{eq_combs}  that each of the above double sums of $w(\n_1) w(\n_2) $  
equals $w(\m)$ times the number of ways of to splitting the multigraph $\G_1(\m)$  into a pair of subgraphs  of  edge multiplicities   $\n_1$ and $\mathcal R(\n_2)$ that are consistent with the imposed source constraints.

 The claimed relation \eqref{eq:switch_reflect}  is implied by the existence of the $1-1$ correspondence between  the two collection of partitions of $\G_1(\m)$, one corresponding to the condition 
 $\partial \n=A$ and the other to $\partial \n=A\Delta \{y,\mathcal y\}$. The bijection is  
 established through the  switch between the two subgraphs along a preselected path over $\gamma(\m)$.   (The existence of such a path is  the essence of the stated assumption for \eqref{eq:switch_reflect}.)
Upon this exchange,  the sources of $\n_1$ and $\mathcal R(n_2)$ are transformed at the endpoints of $\gamma$.  The simultaneous  flips  induce the change
\be
\partial \n \mapsto \partial \n \, \Delta \, \{ y, \mathcal R(y) \}
\ee 
\end{proof}

We are now ready for the main result of this section. 

\begin{proof}[Proof of Theorem\eqref{thm_folded}]  

The RCR  expressions for the two correlation functions are 
\begin{eqnarray} \label{eq_xy}
\langle \sigma_x \sigma_y \rangle      
&=&  \sum_{\substack{\n:\,\E\to \Z_+ \\ \partial \n  = \{x,y\} }}
\frac{
w(\n) }{Z_{\Lambda,\beta}} \, 
\\[2ex]  \notag 
\langle \sigma_x \sigma_{\mathcal {R}(y)} \rangle      
&=& \sum_{\substack{\n:\,\E\to \Z_+ \\ \partial \n=\{x,\mathcal {R}(y)\} }}
\frac{
w(\n)}{Z_{\Lambda,\beta}} \, \, 
\times \id\left[ y \stackrel {\n+\mathcal R(\n) }{ \longleftrightarrow } 
\Lambda_0  \right] \qquad \mbox{} 
\end{eqnarray}
where  the last condition  
is added for free  since for $\n$ with  exactly two sources located at opposite sides of the symmetry plane $\Lambda_0$ the condition  $\mathcal R(y) \stackrel {\n+\mathcal R(\n) }{ \longleftrightarrow } \Lambda_0$ can be taken for granted, and due to the reflection symmetry of $\n +\mathcal R(\n)$  it is equivalent the stated one.

Transforming the last expression  by Lemma~\ref{lem:switch_reflect} one gets:
\be \label{eq_xry}
\langle \sigma_x \sigma_{\mathcal {R}(y)} \rangle      
= 
\sum_{\substack{\n:\,\E\to \Z_+ \\ \partial \n  = \{x,y\} }}
\frac{
w(\n) }{Z_{\Lambda,\beta}}    
\times \id\left[  y \stackrel {\n_1+\mathcal R(\n_2) }{ \longleftrightarrow } 
\Lambda_0  \right]  \,.
\ee
In this form $\langle \sigma_x \sigma_{\mathcal {R}(y)} \rangle $ is given by a partial sum of terms which contribute to   $\langle \sigma_x \sigma_y \rangle $.  
The claimed \eqref{eq_fold_prob} casts the ratio of the sum in \eqref{eq_xry} over \eqref{eq_xy}  in probabilistic terms.  
\end{proof}

\noindent{\bf Remark:}
While the argument was presented for finite domains, 
The infinite volume version of Theorem~\ref{thm_folded} can be deduced through the infinite volume limit of the finite volume states. 

Next are two examples of useful applications of the above considerations.    The added value there is not in the stated inequalities, which are not new, but in the  probabilistic perspective on the remainder terms which the above RCR representation allows to identify.

 \section{Schrader, Messager and Miracle-Sole  inequality}     

The following result was proved by  Schrader~\cite{Sch77}, and independently by Messager and Miracle-Sole~\cite{MesMir77}, in reference to the infinite volume state which results from either the free or the $(+)$ boundary conditions.  (Some additional reflection based inequalities were presented in  \cite{MesMir77, Heg77}).  

\begin{theorem}[The SMMS inequality] \label{thm:SMMS}   For the n.n. ferromagnetic Ising model on $\Z^d$ and any pair of sites $x, y\in \Z^d$
\be\label{eq:SMMS}
\langle \sigma_x \sigma_y\rangle \geq \langle \sigma_x \sigma_{\mathcal R(y)} \rangle \,
 \ee 
for any of the reflection symmetries  $\mathcal R$ of  $\Z^d$ with respect to a hyperplane which does not separate the two sites.
\end{theorem}

By implication, the correlation function $S_2(0,x)$ is decreasing in $|x_j|$ for each of the Cartesian coordinates of $x=\{x_1, ..., x_d\}$.  It is also decreasing in  $|x_1-x_2|$ at fixed values of $\{(x_1+x_2), x_3,..., x_d\}$ and exhibits similar monotonicity   with respect to other choices for coordinate systems with a reflection symmetry.  

Our \eqref{eq_fold_prob} provides an alternative proof of the  SMMS inequality  \eqref{eq:SMMS}.  In addition, it allows to cast the difference in the following form.
\be 
\frac{\langle \sigma_x \sigma_y\rangle - \langle \sigma_x \sigma_{\mathcal R(y)} \rangle }
{\langle \sigma_x \sigma_y\rangle}\,  =  \,  \mathbb P_{\Lambda, \beta}^{\{x,y\}}
\!\!\left( 
x \stackrel{\n+\mathcal R_{\mathbb X}(\n)}{
\not\!\!\!\longleftrightarrow}  \mathbb X  
\right) 
\ee 

One may recognize here the similarity   with a know reflection principle for the hitting probabilities of a simple random walk in $\R^d$ (with optional  absorbing boundary conditions at $\partial \Lambda$).

\section{van Beijeren inequality and the question of roughening transition}

The  possibility of translation symmetry breaking in the classical Ising model in dimensions  $d>2$, was established in Roland Dobrushin's landmark paper~\cite{Dob72}.  His analysis proves  that at low enough temperatures  the Ising model's magnetization  in 
$\Lambda_L =[-L,L]^d$ 
under the  boundary conditions  
\be \label{Dob_bc} 
\mbox{for $ x\in \partial \Lambda_L $}: \quad 
\sigma_x= \begin{cases}  
+1 &   x_d > 0  \\ 
-1 &    x_d < 0   \\
+ 1 &  x_d=0
\end{cases}
\ee 
have magnetization profile which does not vanish in the limit  $L\to \infty$.   

The range of temperatures for which that is proven to hold was extended by  H. van Beijeren~\cite{Bei75} who proved translation symmetry breaking under such conditions for all  $\beta >\beta_c(d-1)$.   His analysis, which is crisp and insightful, is carried  in terms of the Percus representation of duplicated Ising spin systems, and rests on  inequalities of  Lebowitz~\cite{Leb74}.

Among his interesting results is the following statement, in which $\mathcal R$ refers to the reflection with respect to the  hyperplane $\X$ of sites with $x_d=0$.  

\begin{theorem}[van Beijeren \cite{Bei75}]  \label{thm_vanB}
For any Ising spin system in an $\mathcal R$-symmetric domain $\Lambda\subset \Z^d$, 
with an $\mathcal R$-symmetric  nearest neighbor  ferromagnetic pair interaction, and the boundary conditions  \eqref{Dob_bc}:

1) The magnetization at sites in $\Lambda_0 = \Lambda \cap \X$  satisfies
\be  \label{eq_vanB}
\forall x\in \Lambda_0:   \qquad 
\langle \sigma_x \rangle_{\Lambda, \beta}^{\pm,+} \geq \langle \sigma_x \rangle_{\Lambda_0, \beta}^{+}  
\quad \left[ \geq \langle \sigma_0 \rangle_{\Z^{d-1}, \beta}^{+}  \right]
\ee 
where the superscript ${\pm,+}$ refers to the boundary conditions \eqref{Dob_bc}.

2) The following limit exists 
\be
\lim_{L\to \infty} \langle \sigma_x \rangle_{\Lambda_L}^{\pm,+}  =: M_\beta^{\pm; +}(x)  
\ee 
and the limiting magnetization which it defines is concave in $x_d$ for $x_d>0$.
\end{theorem}

Here the double superscripts,  such  $\pm;+$ and  $+;+$,
refer separately to the boundary conditions off and on the symmetry plane, i.e. at $x_d\neq 0$ and $x_d=0$,  listed in that order.

The existence of translation symmetry breaking invites a  number of interesting question   which are still open.  Perhaps the first, which has received some attention but is still open, is whether in $d=3$ dimensions the phenomenon persists  only up to a ``roughening transition'' at some $T_R <T_c$ beyond which $M_\beta^{\pm; +}(x)$ vanishes for all $x$.  
 Also of interest would be to identify what could be the relevant mechanism for such a transition.   

\subsection{RCR representation of the Dobrushin state} \mbox{ } 

Dobrushin state's salient features include the reflection symmetry of the underlying graph,  and the anti-symmetry of the boundary conditions, away from the symmetry plane, and of course the frustration that is induced by the mixed boundary conditions.
 
In Section~\ref{sec:bc} is presented a way to adjust RCR to  frustration through exact  
cancellations,  after which the remaining RCR sum is free of the ``sign-problem''.   The   remaining terms produce a state of the random current which is affected by the conditioning on a suitable  $FF$ (for ``frustration free'')-event.  A key tool for that is a switching argument that is enabled by the symmetry,  or anti-symmetry, found in a doubled system of overlapping currents.  
 
In Section~\ref{sec:reflection} the switching technique was extended to graphs with reflection symmetry, which enables the comparison to be made between two halves of a single current, folded onto itself along the symmetry hyperplane.   

The Dubrushin state may be approached through the natural combination of the  two techniques.

We start by splitting the set  $\Lambda$ and its boundary into the disjoint unions
\begin{eqnarray}
\Lambda &=& \Lambda_0  \sqcup \Lambda_+ \sqcup \Lambda_-  \notag\\[-1ex]  \\[-1ex]  \notag 
\partial \Lambda &=& (\partial\Lambda)_0  \sqcup (\partial\Lambda)_+ \sqcup (\partial\Lambda)_-  
\end{eqnarray} 
with the subscript $0, +,-$ indicating the original set's restriction  to sites  of the corresponding value or sign of $x_d$.

The  arguments used in the proof of  Thorem~\ref{thm:FF} combined with the folding technique of Theorem~\ref{thm_folded} yield the following expression for the ratio of the partition function under the boundary conditions \eqref{Dob_bc} and that under the pure (+) boundary conditions

\begin{theorem}  Under the conditions listed in Theorem~\ref{thm_vanB}:
\be \label{ZZ_Dob}
\frac {Z_{\Lambda,\beta}^{\pm;+}}
{Z_{\Lambda,\beta}^{+;+}}  = 
\mathbb P_{\Lambda,\beta}^{+;+}\left(
(\partial \Lambda)_-   \stackrel{\n+R(\n)}{ \not\!\!\!\longleftrightarrow }  \Lambda_0 \right) 
\ee
and for any $x\in \Lambda_0$:
\begin{eqnarray} \label{eq:1pnt_Dob}
 \langle \sigma_x   \rangle_{\Lambda,\beta}^{\pm; +}\,\, 
 &=&  \mathbb E_{\Lambda,\beta}^{+;+}\left(
 \langle \sigma_x   \rangle_{\Gamma(\n),\beta}^{F; +}
 \,\, \Big | \, \, 
(\partial \Lambda)_-   \stackrel{\n+R(\n)}{ \not\!\!\!\longleftrightarrow } \Lambda_0 \right) 
\\ 
&\geq & \langle \sigma_x   \rangle_{\Z^{d-1},\beta}^{\pm; +} \notag
\end{eqnarray}
where: 
\begin{enumerate} 
\item $\mathbb P_{\Lambda,\beta}^{+;+}$ and $\mathbb E_{\Lambda,\beta}^{+;+}$ denote  the probability and expectation value for the system under the uniformly $(+)$ boundary conditions 
\item $\Gamma(\n)$ is the complement in $\Lambda$ of the set of sites which are $(\n+\mathcal R(\n)) $-connected to $\partial \Lambda\setminus (\partial \Lambda)_0$
\item  the magnetization $\langle \sigma_x   \rangle_{\Gamma(\n),\beta}^{F; +}$ is with respect to boundary conditions which are $(+)$ in the symmetry plane and free above it. 
\end{enumerate} 
\end{theorem}

\begin{figure}[h]               \centering
        a)  \includegraphics[width=0.25\textwidth]{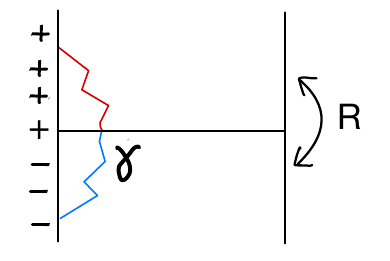}, \quad b)
                  \includegraphics[width=0.5\textwidth]{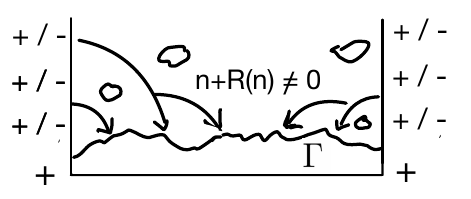}
              \caption{a) The folding correspondence.
              b) The random surface delineating the magnetization bias produced by the antisymmetric boundary conditions. \label{fig:roughening}}
\end{figure}

To help reading  \eqref{eq:1pnt_Dob} it may be pointed out that 
due to the reflection symmetry of $[\n + \mathcal R(\n)]$, for each current configuration the following two conditions are equivalent
\be
(\partial \Lambda)_-   \stackrel{\n+R(\n)}{ \not\!\!\!\longleftrightarrow } \Lambda_0  
\qquad  \Longleftrightarrow \qquad 
(\partial \Lambda)_-   \stackrel{\n+R(\n)}{ \not\!\!\!\longleftrightarrow } (\partial \Lambda)_+  \,.
\ee 
Under the above conditions $\Gamma(\n)$ includes  $\Lambda_0$.  Outside of the symmetry plane  $\Gamma(\n)$'s  boundary consists of 
 an ``outermost null surface'' in terms of the flux of  $(\n+\mathcal R(\n))$  and beyond it possibly also closed surfaces enclosing clusters connected to neither $\Lambda_0$ nor $\partial \Lambda$  (cf. Fig.~\ref{fig:roughening}).   
  
The last inequality, which recovers van Beijeren's bound \eqref{eq_vanB},  is based on the comparison  of   
the Gibbs state within $\Gamma(\n)$ under the stated boundary conditions with its $(d-1)$ dimensional subset $\Lambda_0$ with (+) boundary conditions within the symmetry hyperplane (using Griffiths' monotonicity arguments).   

\begin{proof} 
Starting from the \eqref{ZZ_RCR}  RCR expansion of the partition function under the $(\pm,+)$ boundary conditions, and denoting $\widetilde {\Lambda} :=  \Lambda \setminus \partial \Lambda$ we get the following sequence of transforming equalities
\begin{eqnarray}
Z_{\Lambda,\beta}^{\pm;+} &=& 
\sum_{\substack{\n:\E\to \Z_+ \\ \partial \n \setminus \partial \Lambda = \emptyset}}
w(\n)\,  (-1)^{\mathcal F_-(\n)} \notag \\ 
&=&  
\sum_{\substack{\n:\E\to \Z_+ \\ \partial \n \setminus \partial \Lambda = \emptyset}}
w(\n) (-1)^{\mathcal F_-(\n)}  
\times \id\left[  (\partial \Lambda)_-   \stackrel{\n+R(\n)}{ \not\!\!\!\longleftrightarrow }   \Lambda_0   \right] \notag \\
&=&  
\sum_{\substack{\n:\E\to \Z_+ \\ \partial \n \setminus \partial \Lambda = \emptyset}}
w(\n)   
\times \id\left[  (\partial \Lambda)_-   \stackrel{\n+R(\n)}{ \not\!\!\!\longleftrightarrow }  \Lambda_0   \right] \notag \\ 
&=&  
Z_{\Lambda,\beta}^{+;+} \times
 \mathbb P_{\Lambda,\beta}^{+;+}\left(
 (\partial \Lambda)_-   \stackrel{\n+R(\n)}{ \not\!\!\!\longleftrightarrow }  \Lambda_0 
  \right)
\end{eqnarray}
where the second equality is the pivotal step, in which the omitted terms cancel through a  flip symmetry similar to that discussed in detail within the proof of the single current switching lemma, Lem.~\eqref{lem:switch_reflect}. 
The next step is more  elementary: under the no-connection condition the flux $\mathcal F_(\n)$ is necessarily even.      The last equality is just a rewriting of the sum in the corresponding probabilistic terms.  The last expression is just the  claimed \eqref{ZZ_Dob}.
%

The above reasoning applies equally well to the RCR sum with an added  source within $\Lambda+0$ (this condition is significant since under the added $FF$ condition this source can only be paired with a source at $(\partial \Lambda)_0$, which is not the case otherwise).  
Thus, the above reasoning  also yields:
\begin{eqnarray} \label{eq:1pnt_Dob2}
 \langle \sigma_x   \rangle_{\Lambda,\beta}^{\pm; +}  \,  
 Z_{\Lambda,\beta}^{\pm;+}  &=&  
\sum_{\substack{\n:\E\to \Z_+ \\ \partial \n \setminus \partial \Lambda = \{x\}}}
w(\n)\,  (-1)^{\mathcal F_-(\n)} \\  
  &=&  \sum_{\substack{\n:\E\to \Z_+ \\ \partial \n \setminus \partial \Lambda = \{x\}}}
w(\n)   
\times \id\left[  (\partial \Lambda)_-   \stackrel{\n+R(\n)}{ \not\!\!\!\longleftrightarrow }  \Lambda_0   \right] \notag \\  
  &=&  \sum_{\substack{\n:\E\to \Z_+ \\ \partial \n \setminus \partial \Lambda = \emptyset }}
w(\n)   
\times \id\left[  (\partial \Lambda)_-   \stackrel{\n+R(\n)}{ \not\!\!\!\longleftrightarrow }  \Lambda_0   \right] \langle \sigma_x   \rangle_{\Gamma(\n),\beta}^{F; +}
\notag \\  
&=&  Z_{\Lambda,\beta}^{\pm;+} \, 
\mathbb E_{\Lambda,\beta}^{+;+}\left(
 \langle \sigma_x   \rangle_{\Gamma(\n),\beta}^{F; +}
 \,\, \Big | \, \, 
(\partial \Lambda)_-   \stackrel{\n+R(\n)}{ \not\!\!\!\longleftrightarrow } \Lambda_0 \right) 
\end{eqnarray}
 where in the third equation the effect of the source condition is calculated by conditioning on the values of $\n$ in the complement the above defined set $\Gamma(\n)$.  By the definition of that set, the ration between the conditional sum with versus without the source at $x$ is exactly the indicated mean value of $\sigma_x$ under the free boundary conditions at the boundary of $\Gamma$, except at the level of $\lambda_0$.
 
The last step, like in the previous calculation, is just the rephrasing  of the expression into probabilistic terms of the claimed \eqref{eq:1pnt_Dob}.  
\end{proof}

Discussing the implications of Theorem~\ref{ZZ_Dob}, Ross Graham~\cite{Gra84} established that the roughening transition in three dimensions, corresponds exactly to the temperature at which the boundary of the set $\Gamma$ floats to infinity, in the sense that the probability of it intersecting any given finite region tends to zero as $L\to \infty$.  The question of  existence of this transition strictly below $T_c$ still remains open, although progress has been made on a number or related issues, cf. \cite{She05, GheLub2023}.

\noindent\textbf{Acknowledgements:} 
  I wish to thank Ross Graham, Roberto Fernandez, David Barsky, Joel Lebowitz, Charles Newman, Hugo Duminil-Copin, Vincent Tassion, Simone Warzel, Dima Ioffe, Geoffrey Grimmett,  Mohamed El Hedi Bahri, Ron Peled, Jacob Shapiro, Jakob Bj\"ornberg,  Roland Bauerschmidt, Barry Simon and Frederik Klausen for  stimulating discussions (in roughly  chronological order) of topics related to random currents and their applications.  
I also gratefully acknowledge the  hospitality enjoyed at The Weizmann Institute of Science (Rehovot, Israel), where some of the writing took place.  	

\noindent\textbf{Statements and Declarations:}
The article's author states that  it is free of conflict of interest, and  data sharing is not applicable as no datasets were generated or analysed during the current study.

\noindent\textbf{Note added in proofs:}  
As this paper goes into print the author became aware that, independently of this work, extensions of the random current representation to spin systems with multivariate interactions and frustration, such as present in the DSS setting, are discussed and explored in a one dimensional setting in the parallel work \cite{AGKM25}.   


\end{document}